\newcommand{\ba}{\begin{array}}
\newcommand{\ea}{\end{array}}
\newcommand{\bc}{\begin{center}}
\newcommand{\ec}{\end{center}}

\newcommand{\beqn}[1]{\begin{equation}\label{#1}}
\newcommand{\eeqn}{\end{equation}}
\newcommand{\be}{\begin{equation}}
\newcommand{\ee}{\end{equation}}

\newcommand{\beqnn}{\begin{eqnarray}}
\newcommand{\eeqnn}{\end{eqnarray}}

\newcommand{\col}{{\rm col}}

\newcommand{\diag}{{\rm diag}}

\newcommand{\T}{{\rm T}}
\documentclass[journal,final,twocolumn]{IEEEtran}
\usepackage{amsmath}
\usepackage{graphicx}
\usepackage{amssymb}
\usepackage{cases}
\usepackage{eurosym}
\usepackage{booktabs}
\usepackage{nccmath}
\usepackage{bm}
\usepackage{epstopdf}
\usepackage{algorithm}  
\usepackage[caption=false]{subfig}
\usepackage{algpseudocode}
\usepackage{amsthm}  
\usepackage{mathtools}

\usepackage[nocomma]{optidef}
\theoremstyle{plain}
\newtheorem{thm}{Theorem}
\newtheorem{lem}{Lemma}

\newtheorem{rem}{Remark}
\newtheorem{Def}{Definition}

\newtheorem{asm}{Assumption}

\usepackage[dvipsnames]{xcolor}
\usepackage{hyperref}
\hypersetup{
  colorlinks, linkcolor=red
}
\begin{document}
\title{A Robust Distributed Model Predictive Control Framework for Consensus of Multi-Agent Systems with Input Constraints and Varying Delays}
%\title{Robust distributed model predictive control protocol for consensus of multi-agent systems with varying delays and input constraints}
%{Robust distributed model predictive consensus against information delays}
%{Constrained consensus with resilience to information delays: A robust distributed model predictive control approach}%{Robust distributed model predictive consensus with resilience to information delays} %
% Constrained consensus with resilience to information delays: A robust distributed model predictive control approach
% \title{Distributed Model Predictive Control for Resilient Consensus of Constrained Multi-agent System}
\author{Henglai~Wei,~\IEEEmembership{Member,~IEEE}
        Changxin~Liu,~\IEEEmembership{Member,~IEEE,}
        and~Yang~Shi,~\IEEEmembership{Fellow,~IEEE}% <-this % stops a space
\thanks{H. Wei, C. Liu, Y. Shi are with the Department
of Mechanical Engineering, University of Victoria, Victoria,
BC, V8W 3P6, Canada (e-mail: henglaiwei@uvic.ca; chxliu@uvic.ca; yshi@uvic.ca).}
\thanks{Manuscript received September 1, 2021; revised xx xx, 2021.}}
\markboth{IEEE XXX,~Vol.~14, No.~8, August~2015}%
{Shell \MakeLowercase{\textit{et al.}}: Bare Demo of IEEEtran.cls for IEEE Journals}
\maketitle
%%%%%%%%%%%%%%%%%%%%%%%%%%%%%%%%%%%%%%%%%%%%%%%%%%%%%%%%%%%%%%%%%%%%%%%%%%%%%%
% Currently available protocols either 
\begin{abstract}
This paper studies the consensus problem of general linear discrete-time multi-agent systems (MAS) with input constraints and bounded time-varying communication delays. We propose a robust distributed model predictive control (DMPC) consensus protocol that integrates the offline consensus design with online DMPC optimization to exploit their respective advantages. More precisely, each agent is equipped with an offline consensus protocol, which is \emph{a priori} designed, depending on its immediate neighbors' estimated states. Further, the estimation errors propagated over time due to inexact neighboring information are proved bounded under mild technical assumptions, based on which a robust DMPC strategy is deliberately designed to achieve robust consensus while satisfying input constraints. Moreover, it is shown that, with the suitably designed cost function and constraints, the feasibility of the associated optimization problem can be recursively ensured. We further provide the consensus convergence result of the constrained MAS in the presence of bounded varying delays. Finally, two numerical examples are given to verify the effectiveness of the proposed distributed consensus algorithm.
\end{abstract}
\IEEEpeerreviewmaketitle
%%%%%%%%%%%%%%%%%%%%%%%%%%%%%%%%%%%%%%%%%%%%%%%%%%%%%%%%%%%%%%%%%%%%%%%%%%%%%
       %%%%%%%%%%%%%%%%%%%%%%% Section I %%%%%%%%%%%%%%%%%%%%%%%%%%%%%
%%%%%%%%%%%%%%%%%%%%%%%%%%%%%%%%%%%%%%%%%%%%%%%%%%%%%%%%%%%%%%%%%%%%%%%%%%%%%
\section{Introduction}
\label{whl6-sec:1}
Over recent decades, much research effort has been devoted to the distributed coordination of the multi-agent systems (MAS) since it finds broad applications in numerous areas such as connected vehicles \cite{ju2020distributed}, distributed optimization \cite{nedic2010constrained}, and sensor networks \cite{susca2008monitoring}. As one of the fundamental control problems for the MAS, consensus requires that all agents achieve an agreement of common interest based on the local and neighboring information. To date, many decent algorithms have exhibited impressive results for the MAS; see \cite{olfati2004consensus,ren2005consensus,olfati2007consensus} and the references therein. Nonetheless, reaching consensus becomes more challenging when the MAS are subject to constraints and varying communication delays in real-world applications. A relatively thorough review of related results on the consensus of constrained MAS is presented in the following.

\subsection{Related work} 
The research on the constrained consensus can be found in \cite{ren2008consensus1,lin2013constrained,lin2016distributed,zhao2016global,xie2019global,yang2014global,ong2021consensus}. Several works focus on the MAS with input saturation; see, e.g., \cite{ren2008consensus1,zhao2016global,xie2019global,yang2014global}. These solutions use hyperbolic tangent functions or the low gain feedback technique to construct bounded consensus protocols. More recently, the work \cite{ong2021consensus} presents an output consensus scheme based on the reference governor and the maximal constraint admissible invariable set for the heterogeneous linear MAS with input constraints and switching networks. For the case of MAS subject to state constraints, the authors of \cite{lin2013constrained,lin2016distributed} design the projection algorithms based on the projection operator theory, which attains consensus convergence while meeting state constraints. The work \cite{lin2016distributed} extends the result in \cite{lin2013constrained} to solve the constrained consensus problem of the MAS with switching topology and delayed information. However, explicitly optimizing the consensus performance for the constrained MAS remains a challenge, which motivates the consensus protocol design in this paper. 

\begin{table*}[!ht]
\caption{An overview of existing DMPC-based consensus algorithms.}
\centering
\begin{tabular}{l l l c c }
\toprule
Method  & Constraint types&System dynamics  &Communication delays & Robustness\\
\midrule
\cite{ferrari2009model}  & Unconstrained&Single and double integrators  & --&-- \\
\hline
\cite{muller2012cooperative}  & State and input constraints&Linear and nonlinear systems  & --&--\\
\hline
\cite{zhan2013consensus} & Unconstrained&Single integrators& --&--\\
\hline
\cite{cheng2015distributed} & Input constraints  &Single and double integrators    & --&--\\
\hline
\cite{li2015receding}  & Unconstrained&Linear system& --&--\\
\hline
\cite{li2018receding}  & Input constraints&Linear system  & --&--\\
\hline
\cite{hirche2020distributed}  & State and input constraints& Linear system& --&--\\
\hline
\cite{copp2019distributed}  & Input constraints&Linear and nonlinear systems & --&Min-max DMPC\\
 \hline
\cite{wang2021linear}  & Unconstrained&Linear system & --&--\\
\hline
\textbf{This paper}& Input constraints&Linear system & Bounded varying delays&Tube-based DMPC\\
\bottomrule
\end{tabular}
\label{whl6-tab:1}
\end{table*}

Another powerful class of methods for solving the constrained consensus problem is distributed model predictive control (DMPC) since it has remarkable advantages in handling practical constraints and providing optimal control performance for the MAS \cite{christofides2013distributed}. Previous works along this line include \cite{johansson2006distributed,keviczky2008study,ferrari2009model,muller2012cooperative,zhan2013consensus,cheng2015distributed,li2015neighbor,gao2017distributed,li2018receding,li2015receding,zhan2018distributed,copp2019distributed,hirche2020distributed,wang2020distributed,wang2021linear}. In \cite{ferrari2009model}, a DMPC-based consensus protocol is proposed for the MAS with single and double integrator dynamics over time-varying networks. The geometric properties of the optimal path for each agent are exploited to analyze the consensus convergence. Based on the above technique, a contractive constraint is designed and then incorporated into the DMPC problem for the MAS with double integrator dynamics to achieve consensus. The authors of \cite{zhan2013consensus} present an analytical DMPC solution to the unconstrained average consensus problem and derive the feasible range of the sampling interval for the sampled-data MAS. In \cite{cheng2015distributed}, a DMPC-based consensus framework with adjustable prediction horizon is developed to solve the consensus problem of the discrete-time MAS with double integrator dynamics, input constraints, and switching communication networks. The authors of \cite{li2015neighbor,li2015receding} derive an explicit consensus protocol from the unconstrained DMPC optimization problem for first-order and general linear MAS, where a necessary and sufficient consensus condition is provided. Later, the work \cite{li2018receding} studies the optimal consensus problem for the linear MAS with semi-stable and unstable dynamics and control input constraints. Because of the heavy communication burden and external disturbances, the self-triggered DMPC \cite{zhan2018distributed} and output-feedback DMPC \cite{copp2019distributed} are developed for addressing these practical issues of the constrained MAS, respectively. Recently, the unconstrained consensus problem of the asynchronous MAS with single and double integrator dynamics is solved via DMPC in \cite{wang2020distributed}. The control inputs and consensus states are determined by solving the optimization problem in a distributed manner. The extension of this work for the general linear MAS is discussed in \cite{wang2021linear}. In this scheme, a consensus manifold is introduced such that the final consensus state and input sequence are regarded as augmented decision variables of the DMPC optimization problem. The output consensus problem of the heterogeneous discrete-time MAS subject to state and control input constraints is considered in \cite{hirche2020distributed}. Based on the idea of tracking MPC in \cite{limon2008mpc}, the combination of tracking cost and consensus cost in the overall cost function is beneficial to achieving consensus. It is worth emphasizing that the aforementioned results either focus on simple integrator dynamics \cite{ferrari2009model,cheng2015distributed,wang2020distributed} or only consider unconstrained consensus problems \cite{zhan2013consensus,li2015neighbor,li2015receding}. A detailed comparison between existing DMPC-based consensus algorithms is summarized in TABLE \ref{whl6-tab:1}. However, the above algorithms typically require simultaneous computation and communication at each instant and neglect the transmission delays, rendering them not feasible in many practical scenarios.

\subsection{Contribution} 
In this paper, a robust DMPC-based consensus protocol is proposed to ensure that the general linear constrained MAS reach consensus despite bounded time-varying delays. The main contribution of this paper is threefold. 
\begin{itemize}
\item[1)] \textbf{Distributed consensus protocol for the constrained MAS with delays}: 
We propose a distributed consensus protocol for the general linear MAS subject to input constraints and varying delays. Based on the inverse optimal control \cite{hui20092}, a consensus protocol is firstly designed offline for the unconstrained MAS to achieve global optimality and set stability. Minimizing the gap between the online DMPC input and the predesigned consensus input guarantees the consensus performance of the MAS while satisfying control input constraints. In contrast to existing DMPC-based consensus approaches \cite{zhan2018distributed,gao2017distributed}, the knowledge of the communication topology is explicitly exploited in the design of the offline consensus protocol, which facilitates the analysis of the consensus convergence of MAS. 
\item[2)] \textbf{Robust DMPC handles delay-induced estimation errors}: 
Existing DMPC-based consensus algorithms typically require that all agents simultaneously compute control inputs and exchange optimal predicted states via delay-free wireless networks at each time instant (see, e.g., \cite{li2016distributed,zhan2018distributed}). In this work, we relax these strict communication requirements and consider the non-simultaneous communication and computation, and delay-involved communication networks, in which the optimal predicted state trajectory (i.e., the assumed predicted state trajectory) of each agent broadcast at some previous time instant is used to estimate the current optimal predicted state trajectory. The estimated states would inevitably result in estimation errors that might prevent the MAS from achieving consensus. In this context, we leverage tube-based MPC techniques \cite{chisci2001systems,mayne2005robust} to account for the estimation errors. By bounding the deviation between the assumed and actual predicted states using a properly designed estimation error set, the MAS converges to a neighborhood of the consensus set. %This makes our algorithm applicable to more complex communication environment.
\item[3)] \textbf{Guaranteed feasibility and consensus convergence regardless of delays}: 
Given the robust DMPC-based consensus protocol, conditions for preserving the recursive feasibility are developed. Furthermore, we provide a rigorous theoretical analysis of the consensus convergence for general linear constrained MAS with bounded time-varying communication delays. Finally, two numerical examples are provided to verify the theoretical results.
\end{itemize}

The remainder of this paper is outlined as follows. Section \ref{whl6-sec:2} recalls some basic definitions and describes the problem of interest. Next, Section \ref{whl6-sec:3} introduces the delayed communication networks among the MAS. In Section \ref{whl6-sec:4}, the main results, including the robust DMPC optimization problem and the consensus algorithm, are presented. Section \ref{whl6-sec:5} further provides the corresponding theoretical analysis on the recursive feasibility and the consensus convergence. Two examples are given to illustrate the theoretical results in Section \ref{whl6-sec:6} before concluding remarks are stated in Section \ref{whl6-sec:7}.

\subsection{Notations} The symbols $\mathbb{N}_{\geq m}$ and $\mathbb{N}_{[m,n]}$ denote the sets of the integers greater than or equal to $m$ and integers in the interval $[m,n]$, respectively. $\mathbb{R}^n$ and $\mathbb{R}^{m\times n}$ represent the $n$-dimensional Euclidean space and the set of all $m\times n$ real matrices, respectively. For ${x}\in \mathbb{R}^n$, $\|{x}\|$ denotes the Euclidean norm, $\|{x}\|_{{P}}:=\sqrt{{x}^\text{T}P{x}}$ denotes the weighted Euclidean norm, where $P$ is positive definite. $[{x}_1^\text{T},\dots,{x}_n^\text{T}]^\text{T}$ is written as $\col({x}_1,\dots,{x}_n)$. Given two sets $\mathcal{X},\mathcal{Y}\subseteq \mathbb{R}^n$, the set operation $\mathcal{X}\backslash \mathcal{Y}$ is defined as $\mathcal{X}\backslash \mathcal{Y}:=\{{x}\mid {x}\in\mathcal{X},{x}\notin \mathcal{Y}\}$. The set addition is defined by $\mathcal{X}\oplus\mathcal{Y}:=\{x+y\mid x\in\mathcal{X},y\in\mathcal{Y}\}$, the set subtraction is $\mathcal{X}\ominus\mathcal{Y}:=\{x\in\mathbb{R}^n\mid x\oplus\mathcal{Y}\subseteq\mathcal{X}\}$, and the set multiplication is $K\mathcal{X}:=\{Kx\mid x\in\mathcal{X}\}$, with $K\in\mathbb{R}^{m\times n}$. The Minkowski sum of multiple sets is given by $\mathcal{X}_1\oplus\mathcal{X}_2\oplus\dots\oplus\mathcal{X}_M:=\bigoplus_{i=1}^M\mathcal{X}_i$. The distance between the state ${x}$ and the set $\mathcal{Y}$ is defined as $|{x}|_{\mathcal{Y}}:=\inf_{{y}\in\mathcal{Y}}\|{x}-{y}\|$. $I_M\in\mathbb{R}^{M\times M}$ denotes the identity matrix. $\diag(C_1,C_2,\dots,C_M)$ represents a block diagonal matrix with main diagonal block matrix $C_i$, $i=1,2,\dots,M$. $\bar{\lambda}({P})$ and $\underline{\lambda}(P)$ denote the largest and smallest eigenvalues of the matrix $P$, respectively. For matrices $C\in\mathbb{R}^{m\times n}$, $D\in\mathbb{R}^{p\times q}$, the Kronecker product is denoted by
$$C\otimes D=\left[\begin{array}{ccc}
   c_{11}D & \cdots & c_{1n}D \\
   \vdots & \ddots & \vdots \\
   c_{m1}D & \cdots & c_{mn}D
\end{array}\right].$$
${x}(t)$ denotes the state ${x}$ at time $t$, and ${x}(t+k|t)$ denotes the predicted state at future time $t+k$ determined at time $t$.
%%%%%%%%%%%%%%%%%%%%%%%%%%%%%%%%%%%%%%%%%%%%%%%%%%%%%%%%%%%%%%%%%%%%%%%%%%%%%%
      %%%%%%%%%%%%%%%%%%%%%%% Section II %%%%%%%%%%%%%%%%%%%%%%%%%
%%%%%%%%%%%%%%%%%%%%%%%%%%%%%%%%%%%%%%%%%%%%%%%%%%%%%%%%%%%%%%%%%%%%%%%%%%%%%%
\section{Preliminaries}
\label{whl6-sec:2}
      %%%%%%%%%%%%%%%%%%%%%%% Section 2-1 %%%%%%%%%%%%%%%%%%%%%%%%%
\subsection{Communication graph}
\label{whl6-sec:2-1}
An undirected graph $\mathcal{G}=\{\mathcal{V},\mathcal{E}\}$ is used to describe the information exchange among the MAS, where $\mathcal{V}:=\{1,2,\dots,M\}$ denotes the vertex set and $\mathcal{E}:=\{(i,j)\mid i,j\in\mathcal{V},i\neq j\}$ denotes the edge set. The neighbor set of agent $i$ is denoted by $\mathcal{N}_i:=\{{j}\in\mathcal{V}\mid (i,j)\in\mathcal{E},i\neq j\}$ and the number of agents in $\mathcal{N}_i$ is denoted as $|\mathcal{N}_i|$. Let $\mathcal{A}=[a_{ij}]\in\mathbb{R}^{M\times M}$ be the weighted adjacency matrix of $\mathcal{G}$ with $j\in\mathcal{N}_i$, $a_{ij}=1/|\mathcal{N}_i|$ and $a_{ij}=0$ otherwise; $a_{ii}=\sum_{j=1}^Ma_{ij}=1$ for all $i\neq j$. The Laplacian matrix is $\mathcal{L}=I_M-\mathcal{A}$. 

In this paper, the broadcast communication model is adopted for the MAS, i.e., each agent broadcasts the information to its neighbors via the broadcaster and receives the information from its neighbors via the receiver at each time instant.
      %%%%%%%%%%%%%%%%%%%%%%% Section 2-2 %%%%%%%%%%%%%%%%%%%%%%%%%
\subsection{Set stability}
\label{whl6-sec:2-2}
Consider a general discrete-time system
\begin{equation}\label{whl6-eq:1}
x(t+1)=g(x(t)),\ t\in\mathbb{N}_{\geq 0},
\end{equation}
where $x(t)\in\mathbb{R}^n$, $g:\mathbb{R}^n\to\mathbb{R}^n$ and the solution is denoted by $x(t,x(0))$.

Two definitions of the set stability are introduced next.

\begin{Def} \cite{jiang2002converse}
For the system in \eqref{whl6-eq:1}, a set $\Omega\subset\mathbb{R}^n$ is forward invariant, if $\forall x(0)\in \Omega$, it follows that $x(t,x(0))\in \Omega$, $t\in \mathbb{N}_{\geq 0}$. 
\end{Def}

\begin{Def} {\cite{jiang2002converse}}
The system in \eqref{whl6-eq:1} is asymptotically stable with respect to a forward invariant set $\Pi\subset\mathbb{R}^n$ if the following conditions hold:
\begin{itemize}
\item[1)] Lyapunov stability: for each $\epsilon>0$, there exists some $\delta>0$ such that $|x(0)|_{\Pi}<\delta \Rightarrow |x(t,x(0))|_{\Pi}<\epsilon$, $\forall t\in\mathbb{N}_{\geq 0}$;
\item[2)] Attraction: for $x(0)\in\mathcal{X}\subset\mathbb{R}^n$, $\lim_{t\to\infty}|x(t,x(0))|_{\Pi}=0$, where $\mathcal{X}$ is the region of attraction.
\end{itemize} 
\end{Def}
      %%%%%%%%%%%%%%%%%%%%%%% Section 2-3 %%%%%%%%%%%%%%%%%%%%%%%%%
\subsection{Problem formulation}
\label{whl6-sec:2-3}
Consider a group of $M$ agents that are inter-connected via an undirected graph $\mathcal{G}$. The system model of agent $i$, $i\in \mathcal{V}$ is characterized by 
\begin{equation}\label{whl6-eq:2}
x_i(t+1)=Ax_i(t)+Bu_i(t), \ t\in\mathbb{N}_{\geq 0},
\end{equation} 
where $x_i(t)\in\mathbb{R}^n$ and $u_i(t)\in\mathbb{R}^m$ are the system state and control input, respectively. Agent $i$, $i\in\mathcal{V}$ is subject to control input constraints, i.e., 
\begin{equation}\label{whl6-eq:3}
u_i(t)\in \mathcal{U}_i,\ t\in\mathbb{N}_{\geq 0},
\end{equation}
in which the set $\mathcal{U}_i\subset\mathbb{R}^m$ contains the origin.

\begin{Def} {\cite{you2011network}}
The discrete-time MAS in \eqref{whl6-eq:2} over the fixed graph $\mathcal{G}$ is said to achieve consensus, if for any $x_i(0)$, $i\in\mathcal{V}$, there exists a consensus protocol $u_i(t)=\kappa_i(x_i(t),x_{-i}(t))$ such that $\lim_{t\to+\infty}\|x_i(t)-x_j(t)\|=0$, $j\in\mathcal{N}_i$, where $x_{-i}(t)$ represents the collection of agent $i$'s neighboring states and $\kappa_i:\underbrace{\mathbb{R}^n\times\cdots \times \mathbb{R}^n}_{|\mathcal{N}_i|+1}\to\mathbb{R}^m$.
\end{Def}

In this work, we make the standard assumption on the system dynamics and the communication graph as follows.
\begin{asm}\label{whl6-asm:1}
The pair $(A,B)$ of the MAS in \eqref{whl6-eq:2} is stabilizable, and the associated communication graph $\mathcal{G}$ is connected.
\end{asm}
The connectivity of the graph $\mathcal{G}$ implies that the eigenvalues of $\mathcal{L}$ satisfy $\lambda_i\geq0$, $i=1,2,\dots,M$ \cite{olfati2004consensus}.

Let $\mathbf{x}(t):=\col(x_1(t),x_2(t),\dots,x_M(t))$ and $\mathbf{u}(t):=\col(u_1(t),u_2(t),\dots,u_M(t))$. Then the system in \eqref{whl6-eq:2} can be written in a compact form
\begin{equation}\label{whl6-eq:4}
\mathbf{x}(t+1)=(I_M\otimes A)\mathbf{x}(t)+(I_M\otimes B)\mathbf{u}(t),
\end{equation}
in which $\mathbf{x}(0)=\col(x_1(0),x_2(0),\dots,x_M(0))$ is the initial state of the MAS, the augmented input constraint defined by the Cartesian product of multiple control input constraint sets is $\mathcal{U}:=\mathcal{U}_1\times\mathcal{U}_2\times \dots\times\mathcal{U}_M$. The consensus problem under consideration can be equivalently transformed into a set stabilization problem subject to control input constraints. A consensus set is introduced for the MAS in \eqref{whl6-eq:4} as
\begin{equation}
\mathcal{C}:=\{\mathbf{x}(t)\in\mathbb{R}^{Mn}\mid x_1(t)=x_2(t)=\dots=x_M(t)\},
\end{equation}
where $t\in\mathbb{N}_{\geq 0}$. The MAS in \eqref{whl6-eq:2} achieving consensus implies that the system state $\mathbf{x}(t)$ reaches the consensus set $\mathcal{C}$. It follows that the distance between the state $\mathbf{x}(t)$ and the consensus set $\mathcal{C}$ becomes zero, i.e., $|\mathbf{x}(t)|_{\mathcal{C}}=0$.

\begin{figure*}[!ht]
\centering
\subfloat[Sequential execution \cite{richards2007robust,muller2012cooperative}.]{%
\centering
  \includegraphics[clip,width=0.66\columnwidth]{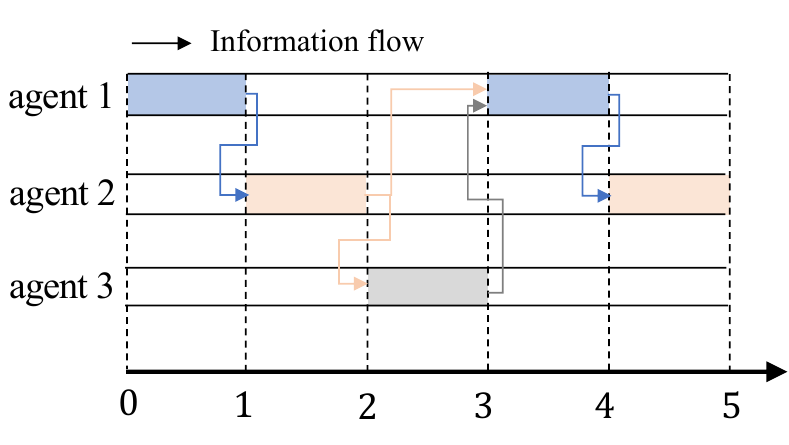}\label{whl6-fig:1-1}%
}
\subfloat[Simultaneous parallel execution \cite{wang2017distributed}.]{%
\centering
  \includegraphics[clip,width=0.66\columnwidth]{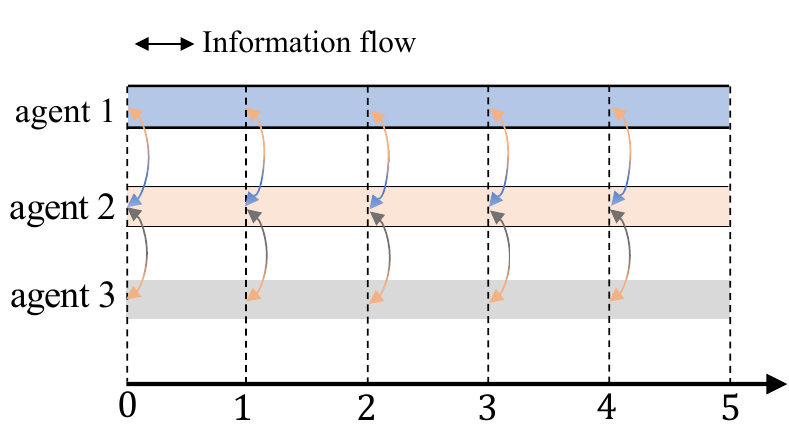}\label{whl6-fig:1-2}%
}
\subfloat[Non-simultaneous parallel execution.]{%
  \includegraphics[clip,width=0.66\columnwidth]{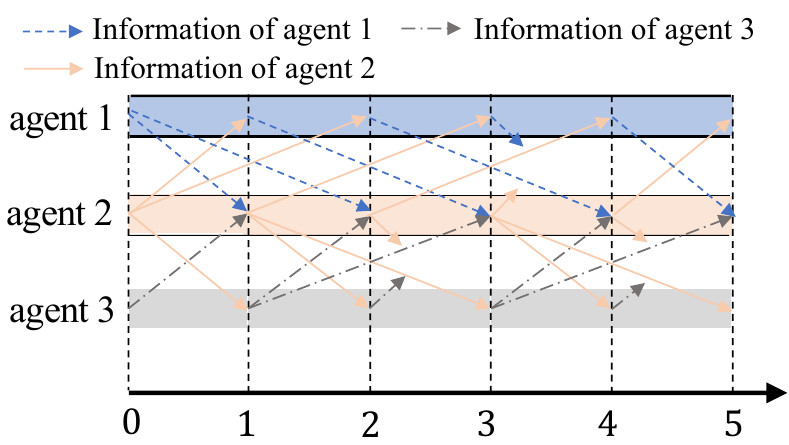}\label{whl6-fig:1-3}%
}
\caption{Three types of implementations of DMPC algorithms for the MAS with $\mathcal{N}_1=\{2\}$, $\mathcal{N}_2=\{1,3\}$, $\mathcal{N}_3=\{2\}$.}\label{whl6-fig:1}
\end{figure*}

The following lemma (\cite[Theorem 1]{li2009consensus}) provides a sufficient and necessary condition for the MAS to reach consensus. 
\begin{lem}\label{whl6-lem:1}
For agent $i$, $i\in\mathcal{V}$ in \eqref{whl6-eq:2} over a fixed communication graph $\mathcal{G}$, the consensus of the MAS can be achieved if and only if there exists a predesigned consensus gain $K\in\mathbb{R}^{m\times n}$ such that $\rho(A+\lambda_iBK)<1$, where $\lambda_i$, $i=2,3,\dots,M$, are the nonzero eigenvalues of the Laplacian matrix $\mathcal{L}$.
\end{lem}

For the recursive feasibility of the optimization problem, we make the following assumption on the predesigned consensus gain $K$.
\begin{asm}\label{whl6-asm:2}
There exist a feedback control matrix $K$ and a forward invariant set $\Omega$ such that: 1) $\rho(A+BK)<1$; 2) $K\sum_{j\in\mathcal{N}_i}a_{ij}(x_i-x_j)\in\mathcal{U}_i$, when the states $x_i\in\Omega$, $i\in\mathcal{V}$ and $x_j\in\Omega$, $j\in\mathcal{N}_i$.
\end{asm}
The forward invariant set (i.e., the terminal set) will be given in Section \ref{whl6-sec:4-2}.

The control objective is to design a distributed consensus protocol for the MAS in \eqref{whl6-eq:2} such that all agents attain:
\begin{itemize}
\item[1)] \emph{Agreement}: For agents $i$, $j\in\mathcal{V}$, the following condition 
\begin{equation}
{\lim_{t\to\infty}\|x_i(t)-x_j(t)\|=0,\ j\in\mathcal{N}_i},
\end{equation}
holds, which is equivalent to $\lim_{t\to\infty}|\mathbf{x}(t)|_{\mathcal{C}}=0$.
\item[2)] \emph{Constraint satisfaction}: The system $i$, $i\in\mathcal{V}$
\begin{equation}
x_i(t+1)=Ax_i(t)+Bu_i(t), 
\end{equation}
satisfies control input constraints in \eqref{whl6-eq:3} for all $t\in\mathbb{N}_{\geq 0}$.
\end{itemize}
%%%%%%%%%%%%%%%%%%%%%%%%%%%%%%%%%%%%%%%%%%%%%%%%%%%%%%%%%%%%%%%%%%%%%%%%%%%%%%
      %%%%%%%%%%%%%%%%%%%%%%% Section 3 %%%%%%%%%%%%%%%%%%%%%%%%%
%%%%%%%%%%%%%%%%%%%%%%%%%%%%%%%%%%%%%%%%%%%%%%%%%%%%%%%%%%%%%%%%%%%%%%%%%%%%%%
\section{Delayed communication among MAS}
\label{whl6-sec:3}
This section introduces the assumed information broadcast among the MAS over the delay-involved wireless networks.

Most existing DMPC algorithms for MAS only consider delay-free networks, and some of them even require stricter communication settings (e.g., simultaneous communication and computation). Here, three types of DMPC schemes are presented in Fig.~\ref{whl6-fig:1}, including  
\begin{itemize}
{\item[(a)] \emph{Sequential execution:} As shown in Fig.~\ref{whl6-fig:1-1}, agents solve local optimization problems and communicate with their neighbors sequentially. Only one agent calculates the optimal control inputs and broadcasts the optimal predicted state sequence at each time instant.
\item[(b)] \emph{Parallel execution with simultaneous computation and communication:} As shown in Fig.~\ref{whl6-fig:1-2}, at each time instant, all agents simultaneously receive neighbors' information, solve the optimization problems, and broadcast information to their neighbors.
\item[(c)] \emph{Parallel execution with non-simultaneous computation and communication:} As depicted in Fig.~\ref{whl6-fig:1-3}, at each time instant, agents broadcast the predicted states. However, these broadcast predicted states might be used by their neighbors a few time instants later due to the time-varying communication delays.}
\end{itemize}

In contrast to sequential algorithms \cite{richards2007robust,muller2012cooperative}, parallel DMPC algorithms \cite{li2017robust} provide a more efficient coordination solution for constrained MAS in terms of both computation and communication. Most of the existing results assume that the communication among agents is perfect, and the information can be synchronously exchanged as shown in Fig. \ref{whl6-fig:1-2} (e.g., see \cite{wang2017distributed}). However, it is unrealistic for agents to calculate the optimal predicted states while exchanging them simultaneously and instantaneously. Especially, the communication delays are unavoidable during the information exchange for the MAS in many practical scenarios, which may result in the non-simultaneous parallel algorithm execution as shown in Fig.~\ref{whl6-fig:1-3}. That is, the optimal predicted state sequence broadcast at the previous time instant $t'$, $t'\in\mathbb{N}_{[t-\bar{\tau},t-1]}$, $\bar{\tau}\in\mathbb{N}_{\geq 1}$ (the assumed predicted state sequence) is used to estimate the current optimal predicted state sequence at $t$, $t\in\mathbb{N}_{\geq 0}$. Notably, the estimation errors are regarded as external disturbances for the MAS. 

The assumed states of agent $i$ at time $t$ are constructed based on the previous optimal predicted states at time $t'$, i.e., 
\begin{equation}\label{whl6-eq:9}
\begin{aligned}
&\hat{x}_i(t+k|t)=\\
&\begin{dcases} 
      x_i^*(t+k|t'), & k\in\mathbb{N}_{[0,N']}, \\
      A\hat{x}_i(t+k-1|t)+B\hat{u}(t+k-1|t),& k\in\mathbb{N}_{(N',N]},
\end{dcases}
\end{aligned}
\end{equation}
in which $t'\in\mathbb{N}_{[t-\bar{\tau},t-1]}$, $N'=t'+N-t$, $N$ is the prediction horizon, $x_i^*(\cdot|t')$ represents the optimal state and $\hat{u}_i(t+k|t)=K\sum_{j\in\mathcal{N}_i}a_{ij}\big({x}_i(t+k|t)-\hat{x}_j(t+k|t)\big)$ with $k\in\mathbb{N}_{(N',N]}$. Note that states $x_j^*(\cdot|t')$, $j\in\mathcal{N}_i$, $t'\in\mathbb{N}_{[t-\bar{\tau},t-1]}$ are available for agent $i$ at time $t$. Let $\hat{\bm{x}}_i(t):=\{\hat{x}_i(t+k|t)\}$, $t\in\mathbb{N}_{\geq 0}$, $k\in\mathbb{N}_{[0,N]}$ be the assumed predicted state sequence of agent $i$, $i\in\mathcal{V}$ at time $t$ hereafter. 

We make the following assumption on the communication delays among the MAS.
\begin{asm}\label{whl6-asm:3}
For agent $i$, $i\in\mathcal{V}$, the time-varying communication delays $\tau(t)\in\mathbb{N}_{\geq 0}$ between agent $i$ and $j$, $j\in\mathcal{N}_i$, satisfy $1\leq \tau(t)\leq \bar{\tau}<N$, with $\tau(0)=0$ and $\bar{\tau}$ being the largest communication delay.
\end{asm}
The upper bound of the communication delays ensures that the previously broadcast assumed predicted states can be used to estimate the actually optimal predicted states. For example, at time $t$, agent $i$ can receive the assumed predicted state sequence $x_i^*(\cdot|t')$ from its neighbor $j$, $j\in\mathcal{N}_i$ broadcast at time $t'\in\mathbb{N}_{[t-\bar{\tau},t-1]}$. Then, the received information will be adopted as in \eqref{whl6-eq:9} to construct the current assumed predicted state sequence at time $t$. %On the other hand, the case of unbounded delays can be modeled as packet dropouts in networked control systems \cite{li2013network} or denial-of-service attacks in cyber-physical systems \cite{sun2019resilient}, which is beyond the scope of current work.

\begin{rem}
Many consensus algorithms have been proposed for the MAS with delays. For example, the conditions based on special stochastic matrix properties are derived for the discrete-time MAS, e.g., \cite{xiao2008asynchronous}, and linear matrix inequality conditions for the continuous-time MAS, e.g., \cite{sun2009consensus}. In contrast, the previously broadcast predicted states can be used to design the distributed consensus protocol owning to the prediction mechanism of MPC. Notably, in existing DMPC results (e.g., \cite{li2016distributed,zhan2018distributed}), the assumed predicted state sequence (one-step-ahead predicted state sequence) is generally used to estimate the current optimal state sequence, which can be regarded as a special case of the assumed predicted state sequence of the MAS with varying delays in \eqref{whl6-eq:9}, with $\tau(t)=1$.
\end{rem}

Note that restricting the deviation between the current optimal predicted states and the assumed predicted states is necessary for the MAS to reach consensus. In particular, the current optimal predicted state of agent $i$, $i\in\mathcal{V}$ is supposed to lie in a bounded neighborhood (i.e., the estimation error set $\Delta$) of the assumed predicted state, 
\begin{equation}\label{whl6-eq:10}
x_i^*(t+k|t)\in \hat{x}_i(t+k|t)\oplus \Delta,\ k\in\mathbb{N}_{[0,N]},
\end{equation}
where $\Delta:=\{\delta\in\mathbb{R}^n\mid\|\delta\|<\eta\}$ is symmetric and contains the origin with $\eta>0$. 

We consider a scenario where agents are subject to bounded varying communication delays. In this context, the control objective presented in Section~\ref{whl6-sec:2-3} can be restated as follows.
\begin{itemize}
\item[1)] \emph{Robust agreement}: For agents $i$, $j\in\mathcal{V}$, applying the proposed distributed consensus protocol, then
\begin{equation}
{\lim_{t\to\infty}\|x_i(t)-{x}_j(t)\|\leq \gamma,\ j\in\mathcal{N}_i},
\end{equation}
holds, where the invariant set $\mathcal{R}_i^\infty:=\{x\in\mathbb{R}^n\mid \|x\|\leq \gamma\}$ and $\gamma>0$.
\item[2)] \emph{Constraint satisfaction}: The system $i$, $i\in\mathcal{V}$
\begin{equation}
x_i(t+1)=Ax_i(t)+Bu_i(t), 
\end{equation}
satisfies control input constraints in \eqref{whl6-eq:3} for all $t\in\mathbb{N}_{\geq 0}$.
\end{itemize}

Note that using \textbf{Assumption \ref{whl6-asm:2}} allows us to get a minimal robust invariant set $\mathcal{R}_i^{k}$ as $k\to \infty$ \cite{mayne2005robust}. 

%%%%%%%%%%%%%%%%%%%%%%%%%%%%%%%%%%%%%%%%%%%%%%%%%%%%%%%%%%%%%%%%%%%%%%%%%%%%%%
      %%%%%%%%%%%%%%%%%%%%%%% Section 4 %%%%%%%%%%%%%%%%%%%%%%%%%
%%%%%%%%%%%%%%%%%%%%%%%%%%%%%%%%%%%%%%%%%%%%%%%%%%%%%%%%%%%%%%%%%%%%%%%%%%%%%%
\section{Robust DMPC-based consensus protocol}
\label{whl6-sec:4}
The robust DMPC-based consensus scheme for the MAS is illustrated in Fig.~\ref{whl6-fig:2}, which mainly consists of five parts: the controlled system, the robust DMPC controller, the broadcaster, the receiver and wireless communication networks.
\begin{figure}[ht]
\centering
\includegraphics[width=0.9\columnwidth]{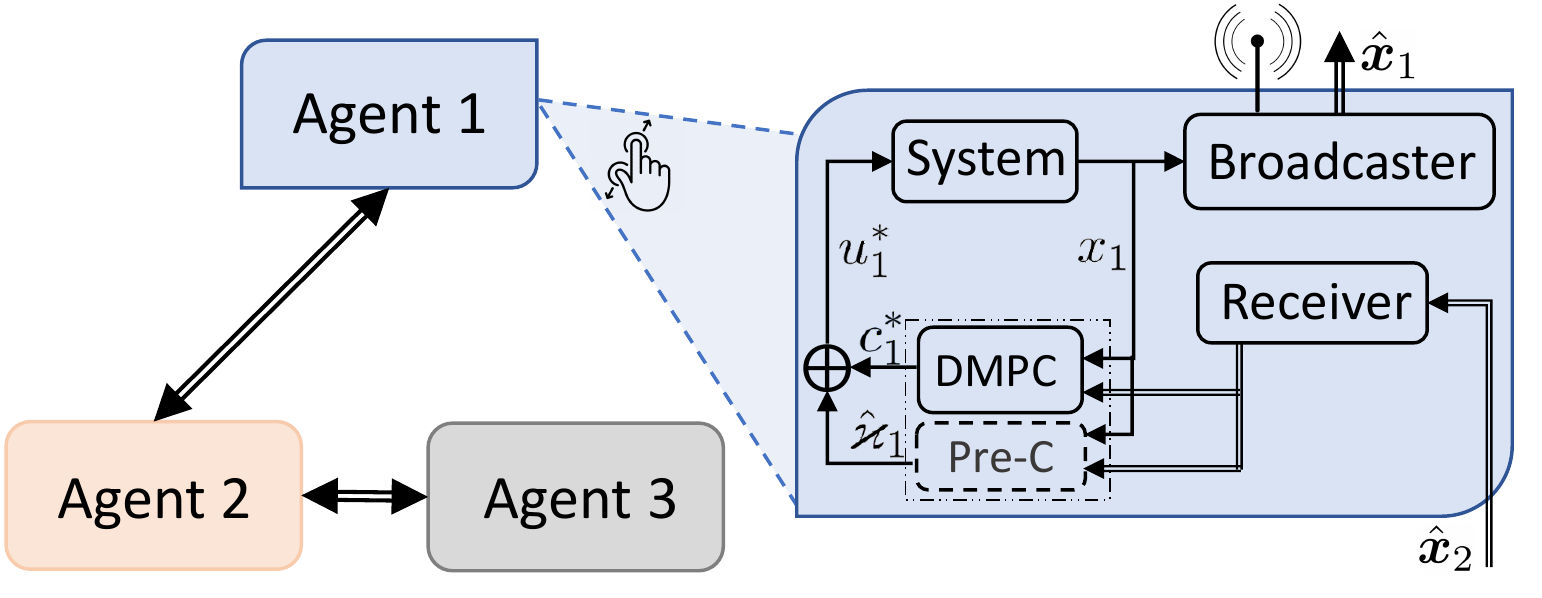}
\caption{Illustration of the robust DMPC-based consensus scheme for the MAS with $\mathcal{N}_1=\{2\}$, $\mathcal{N}_2=\{1, 3\}$, $\mathcal{N}_3=\{2\}$. At each time instant, the \textsf{Broadcaster} of agent $i$, $i\in\mathcal{V}$ broadcasts the assumed predicted state sequence $\hat{\bm{x}}_i$ and the \textsf{Receiver} receives the assumed predicted state sequence $\hat{\bm{x}}_j$ from agent $j$, $j\in\mathcal{N}_i$ via wireless communication networks.}
\label{whl6-fig:2}
\end{figure}

\subsection{Predesigned consensus protocol for the MAS without delays}
\label{whl6-sec:4-1}
At time $t$, agent $i$ calculates control inputs, broadcasts predicted system states $\hat{\bm{x}}_i(t)$ to its neighbor $j$, $j\in\mathcal{N}_i$ via the network $\mathcal{G}$ and updates system states $x_i(t)$. The consensus protocol for agent $i$ is predesigned based on the relative states between agent $i$ and $i$'s neighbors under the assumption that the wireless network is delay-free, i.e.,
\begin{equation}\label{whl6-eq:6}
\begin{aligned}
u_i(t)&=\kappa_i(x_i(t),x_{-i}(t))\\
&=K\sum_{j\in\mathcal{N}_i}a_{ij}\big(x_i(t)-x_j(t)\big)+c_i(t),
\end{aligned}
\end{equation}
where $K\in\mathbb{R}^{m\times n}$, $\varkappa_i(t):=-K\sum_{j\in\mathcal{N}_i}a_{ij}\big(x_i(t)-x_j(t)\big)$ is the predesigned consensus input and $c_i(t)$ is the control decision variable of the robust DMPC optimization problem that will be designed in Section~\ref{whl6-sec:4-2}.

\begin{rem}
Note that the proposed consensus protocol \eqref{whl6-eq:6} is motivated by the well-known ``pre-stabilizing'' control method for the stabilization problem of the single system in \cite{chisci2001systems,goulart2006optimization}, in which the MPC control input is given by $u(t)=Kx(t)+c(t)$. The feedback gain $K$ is chosen offline for the unconstrained system to achieve desired properties such as LQR optimality and stability, and $c(t)$ is calculated online. In this work, an inverse optimal consensus protocol $\varkappa_i(t)$ is designed offline for the unconstrained MAS in \eqref{whl6-eq:2} and $c_i(t)$ is determined via solving the online optimization problem. The resulting consensus protocol achieves suboptimal performance while guaranteeing the satisfaction of control input constraints. In contrast to the existing DMPC methods in \cite{li2016distributed,zhan2018distributed,gao2017distributed}, the knowledge of the communication topology is explicitly exploited to construct the predesigned consensus protocol, which facilitates the consensus convergence analysis of the constrained MAS.
\end{rem}

\subsection{The predesigned consensus protocol for the MAS with delays}
\label{whl6-sec:4-2}
We now present the predesigned consensus protocol for the MAS with time-varying communication delays. Based on the assumed predicted states of neighbors, the consensus protocol in \eqref{whl6-eq:6} then becomes
\begin{equation}\label{whl6-eq:12}
u_i(t)=K\sum_{j\in\mathcal{N}_i}a_{ij}\big(x_i(t)-\hat{x}_j(t)\big)+c_i(t),
\end{equation}
and the MAS in \eqref{whl6-eq:2} under the consensus protocol in \eqref{whl6-eq:12} can be written as
\begin{equation}\label{whl6-eq:13}
\begin{aligned}
&x_i(t+1)\\
=&Ax_i(t)+BK\sum_{j\in\mathcal{N}_i}a_{ij}\big(x_i(t)-\hat{x}_j(t)\big)+Bc_i(t).
\end{aligned}
\end{equation}

The estimation errors of agent $i$, $i\in\mathcal{V}$ are defined as $w_i(t)=\sum_{j\in\mathcal{N}_i}a_{ij}\big(\hat{x}_j(t)-{x}_j(t)\big)$ and satisfy $w_i(t)\in\mathcal{W}:=\{w\in\mathbb{R}^n\mid w\in\bigoplus_{j\in\mathcal{N}_i}a_{ij}\Delta\}$. In particular, the estimation error will be treated as the disturbance in the following. Due to $\sum_{j\in\mathcal{N}_i}a_{ij}=1$, it holds that $\mathcal{W}=\Delta$.

Under the consensus protocol \eqref{whl6-eq:12}, the closed-loop system in \eqref{whl6-eq:13} becomes 
\begin{equation*}%\label{whl6-eq:14}
\begin{aligned}
&x_i(t+1)\\
=&Ax_i(t)+B\big(K\sum_{j\in\mathcal{N}_i}a_{ij}\big(x_i(t)-{x}_j(t)\big)+c_i(t)\big)-BKw_i(t).
\end{aligned}
\end{equation*}
\begin{rem}
Existing DMPC-based consensus methods (e.g., \cite{li2016distributed,zhan2018distributed}) do not consider the estimation errors explicitly, which may not be applicable for practical MAS. Because the estimation errors may prevent MAS from achieving consensus, inspired by the tube-based MPC for the uncertain linear system in \cite{mayne2005robust}, one might choose to impose the estimation error set to bound the estimation errors that are regarded as external disturbances in this work. 
\end{rem}

\subsection{Design of robust DMPC for consensus}
\label{whl6-sec:4-3}
In this work, the terminal set for the MAS in \eqref{whl6-eq:4} is defined as $\mathbf{X}^f:=\{\mathbf{x}\in\mathbb{R}^{Mn}\mid\|\mathbf{x}\|_\mathbb{S}\leq\epsilon^2\}$, in which $\mathbb{S}=\sigma\mathcal{L}\otimes S$, with the positive semi-definite weighting matrix $S\in\mathbb{R}^{n\times n}$ and the positive constants $\epsilon$ and $\sigma$. The matrix $S$ can be chosen following the method provided in Lemmas $6$ and $7$ in \cite{li2018receding}. The terminal set $\mathcal{X}_i^f$ for agent $i$, $i\in\mathcal{V}$ is given by
\begin{equation}
{\mathcal{X}_i^f:=\{x_i\in\mathbb{R}^n\mid \sum_{j\in\mathcal{N}_i}a_{ij}x_i^{\text{T}}S(x_i-{x}_j)\leq\epsilon^2/M\}}.
\end{equation}

The cost function $J_i(\cdot)$ for agent $i$, $i\in\mathcal{V}$ is defined by 
\begin{equation}\label{whl6-eq:16}
J_i(\bm{c}_i(t)):=\sum_{k=0}^{N-1}\|{c}_i(t+k|t)\|^2_{P_i},
\end{equation}
where $\bm{c}_i(t):=\{{c}_i(t|t),{c}_i(t+1|t),\dots,{c}_i(t+N-1|t)\}$ denotes the control sequence generated at time instant $t$ and the matrix $P_i$ is positive definite. Note that the suboptimal performance can be guaranteed by minimizing the cost function in \eqref{whl6-eq:16} while ensuring the satisfaction of control input constraints.

At time $t$, given the current state $x_i(t)$ of agent $i$, $i\in \mathcal{V}$ and neighbors' assumed predicted state sequence $\hat{\bm{x}}_j(t)$, $j\in\mathcal{N}_i$, the robust DMPC optimization problem $\mathcal{P}_i$ is formulated as
\begin{fleqn}
\begin{subequations}\label{whl6-eq:17}
\begin{alignat}{2}
\min_{\bm{c}_i(t)}\ & {J}_{i}(\bm{c}_i(t)) \notag\\
\text{s.t.}\ & x_i(t|t)=x_i(t),\label{whl6-eq:17a}\\
&x_i(t+k+1|t)=Ax_i(t+k|t)+Bu_i(t+k|t),\hskip -0.5em\\
&u_i(t+k|t)=\hat{\varkappa}_i(t+k|t)+c_i(t+k|t),\\
&{{u}}_i(t+k|t)\in {\mathcal{U}}_i,\label{whl6-eq:17d}\\
&x_i(t+k|t)\in\hat{x}_i(t+k|t)\oplus \Delta, \label{whl6-eq:17e}\\
&x_i(t+N|t)\in\mathcal{X}_i^f,\label{whl6-eq:17f}
\end{alignat}
\end{subequations}
\end{fleqn}
in which $k\in \mathbb{N}_{[0,N-1]}$ and $\hat{\varkappa}_i(t+k|t):=K\sum_{j\in\mathcal{N}_i}a_{ij}\big({x}_i(t+k|t)-\hat{x}_j(t+k|t)\big)$. Note that the constraint \eqref{whl6-eq:17e} bounds the optimal predicted state sequence $x_i^*(t+k|t)$ within a predesigned tube centered along the assumed state $\hat{x}_i(t + k|t)$, $k\in\mathbb{N}_{[0,N-1]}$. Let $\bm{c}_i^*(t):=\{c_i^*(t|t),\dots,c_i^*(t+N-1|t)\}$ be the optimal solution to the robust DMPC problem $\mathcal{P}_i$ at time $t$. We then have the corresponding optimal control input 
\begin{equation}\label{whl6-eq:20}
\begin{aligned}
&{u}_i^*(t+k|t)\\
=&K\sum_{j\in\mathcal{N}_i}a_{ij}\big({x}_i^*(t+k|t)-\hat{x}_j(t+k|t)\big)+c_i^*(t+k|t),
\end{aligned}
\end{equation}
where $k\in\mathbb{N}_{[0,N-1]}$, and the optimal control sequence at time $t$ is $\bm{u}_i^*(t):=\{{u}_i^*(t|t),u_i^*(t+1|t),\dots,u_i(t+N-1|t)\}$. The resultant optimal predicted state satisfies
\begin{equation}\label{whl6-eq:21}
x_i^*(t+k+1|t)=Ax_i^*(t+k|t)+Bu_i^*(t+k|t),
\end{equation}
where $k\in\mathbb{N}_{[0,N-1]}$, $x_i^*(t|t)=x_i(t)$, and the optimal predicted state sequence is denoted by $\bm{x}_i^*(t):=\{x_i^*(t|t),x_i^*(t+1|t),\dots,x_i^*(t+N|t)\}$.  

At time $t$, applying $u_i^*(t|t)$ to the system in \eqref{whl6-eq:2} yields the following closed-loop system
\begin{equation}\label{whl6-eq:22}
\begin{aligned}
&x_i(t+1)\\
=&Ax_i(t)+Bu_i^*(t|t)\\
=&Ax_i(t)+B\big(K\sum_{j\in\mathcal{N}_i}a_{ij}\big({x}_i^*(t|t)-\hat{x}_j(t|t)\big)+c_i^*(t|t)\big).
\end{aligned}
\end{equation}
And according to \eqref{whl6-eq:21}, one has $x_i(t+1)=x_i^*(t+1|t)$.

\begin{rem}
Note that many DMPC algorithms have been developed for the formation stabilization control problem of the MAS, e.g., \cite{dunbar2006distributed}. However, they cannot be easily extended to consensus problems since the optimal cost function may not be directly used as Lyapunov function \cite{ferrari2009model,muller2012cooperative}. In this case, the major difficulty when analyzing the consensus convergence of the MAS is designing some specific conditions such that the Lyapunov theorem holds \cite{li2016distributed,zhan2018distributed}. Alternatively, this paper exploits the knowledge of the communication topology to design an offline optimal consensus protocol $\sum_{j\in\mathcal{N}_i}a_{ij}K(x_i-x_j)$ for the unconstrained MAS. Then the difference $c_i$ between the unconstrained control input $\hat{\varkappa}_i$ and the control input $u_i$ is minimized, which will be proved to be a vanishing input in \textbf{Lemma \ref{whl6-lem:3}}. 
\end{rem}
% \subsection{Robust DMPC-based consensus algorithm}
% \label{whl6-sec:4-4}
The proposed robust DMPC-based consensus algorithm is presented as follows. 

\begin{algorithm}[H]
\caption{Robust DMPC-based consensus algorithm}\label{whl6-alg:1}
\begin{algorithmic}[1]
\Require The prediction horizon $N$, the predesigned consensus gain $K$, the terminal set $\mathcal{X}_i^f$, the estimation error set $\Delta$, the weighting matrix $P_i$, the initial state $x_i(0)$, the assumed state sequence $\hat{\bm{x}}_i(0):=\{x_i(0),\dots,A^{N-1}x_i(0)\}$, and other parameters. 

{\hspace{-3.8em}\textbf{Online Procedure:}}
\State Broadcast the assumed predicted state sequence $\hat{\bm{x}}_i(0)$ to its neighbors $j$, $j\in\mathcal{N}_i$;\label{step1} 
\While{for agent $i$, the control is not stopped}
\State Measure the current system state $x_i(t)$; 
\State Receive the information $\hat{\bm{x}}_j(t)$, $j\in \mathcal{N}_i$ defined in \eqref{whl6-eq:9};\label{step2}
\State {Solve the problem $\mathcal{P}_i$ to generate $\bm{u}_{i}^{*}(t)$ and $\bm{x}_i^*(t)$;} \label{step3}  
\State Broadcast $\hat{\bm{x}}_i(t)$ to agent $j$, $j\in\mathcal{N}_i$;
\State Apply the control input ${u}_i^*(t|t)$ to agent $i$; 
\State ${t}={t}+1$;
\EndWhile\label{endwhile}
\end{algorithmic}
\end{algorithm}

Note that in Step $4$ of \textbf{Algorithm \ref{whl6-alg:1}}, agent $i$ collects the assumed predicted state sequence from its neighbors $j$, $j\in\mathcal{N}_i$. The predicted state sequences of neighbors broadcast at the common time instant $t'$, $t'\in\mathbb{N}_{[t-\bar{\tau},t-1]}$ as in \eqref{whl6-eq:9}, instead of the most recently received predicted state sequences, will be used in Step $5$ to generate the control input $\bm{u}^*_i(t)$. 

\section{Theoretical analysis}
\label{whl6-sec:5}
This section provides the theoretical analysis of the recursive feasibility and the consensus convergence of the MAS. Before presenting the main results, the relationship between the theoretical results is illustrated in Fig.~\ref{whl6-fig:3a}.
\begin{figure}[!ht]
\centering
\includegraphics[width=0.98\columnwidth]{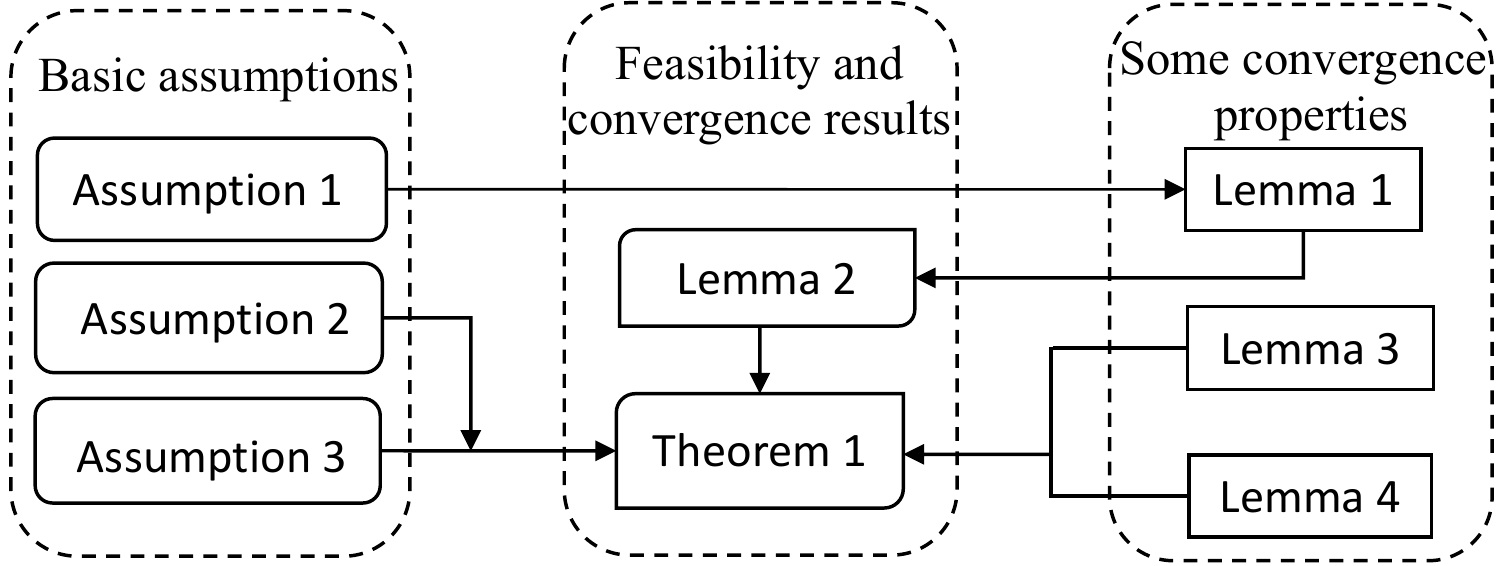}
\caption{The roadmap of the theoretical results.}
\label{whl6-fig:3a}
\end{figure}

\subsection{Recursive feasibility result}
\label{whl6-sec:5-1}
This subsection discusses the recursive feasibility of the robust DMPC optimization problem. It is assumed that an initial feasible solution to the optimization problem $\mathcal{P}_i$ exists with the given initial state $x_i(t)$, $t\in\mathbb{N}_{\geq 0}$. 

Based on the optimal predicted control sequence at time $t''$, $t''\in\mathbb{N}_{[t+1-\bar{\tau},t]}$, a candidate input sequence $\tilde{\bm{c}}_i(t+1)$ at time $t+1$ is created by dropping the first $t-t''$ input element and appending final $t-t''$ terminal zero element of the optimal control sequence at $t''$, i.e.,
\begin{equation}\label{whl6-eq:23}
\tilde{c}_i(t+k|t+1)=\begin{dcases}
c_i^*(t+k|t''),& k\in\mathbb{N}_{[1,t''+N-t)},\\
0,& k\in\mathbb{N}_{[t''+N-t,N]},
\end{dcases}
\end{equation}
and the corresponding control input sequence $\tilde{\bm{u}}_i(t+1):=\{\tilde{u}_i(t+1|t+1),\tilde{u}_i(t+2|t+1),\dots,\tilde{u}_i(t+1+N-1|t+1)\}$ for agent $i$ is then constructed as
\begin{equation}\label{whl6-eq:24}
\begin{aligned}
&\tilde{u}_i(t+k|t+1)\\
=&K\sum_{j\in\mathcal{N}_i}a_{ij}\big(\tilde{x}_i(t+k|t+1)-\hat{x}_j(t+k|t+1)\big)\\
&+\tilde{c}_i(t+k|t+1),
\end{aligned}
\end{equation}
where $k\in\mathbb{N}_{[1,N]}$ and the corresponding system state $\tilde{x}_i(t+1+k|t+1)$ satisfies the following difference equation
\begin{equation}\label{whl6-eq:25}
\begin{aligned}
&\tilde{x}_i(t+1+k+1|t+1)\\
=&A\tilde{x}_i(t+1+k|t+1)+B\tilde{u}_i(t+1+k|t+1),
\end{aligned}
\end{equation}
with the initial condition $\tilde{x}_i(t+1|t+1)=x_i(t+1)$.

In the following lemma, we provide the conditions under which the control input sequence defined in \eqref{whl6-eq:23} and \eqref{whl6-eq:24} will be a feasible solution to the robust DMPC optimization problem $\mathcal{P}_i$ at time $t+1$.

\begin{lem}\label{whl6-lem:2}
For the MAS in \eqref{whl6-eq:2}, suppose \textbf{Assumptions~\ref{whl6-asm:1}, \ref{whl6-asm:2}, \ref{whl6-asm:3}} and the condition in \textbf{Lemma~\ref{whl6-lem:1}} hold. If there exist the predesigned consensus gain $K$, the maximum delay $\bar{\tau}$ and the prediction horizon $N$, such that the conditions
\begin{equation}\label{whl6-eq:21a}
\begin{aligned}
&\max_{k\in\mathbb{N}_{[1,N'-1]}}\big\{\rho(\sum_{s=0}^{k-1}A_K^{k-1-s}BK+A_K^k)\big\}\leq 1,\ \text{and}\\
&\max_{k\in\mathbb{N}_{[N',N-1]}}\big\{\rho(\sum_{s=0}^{N'-1}A_K^{k-1-s}BK+A_K^k)\big\}\leq 1,
\end{aligned}
\end{equation} 
hold, with $\rho(\cdot)$ denoting the spectral radius, $A_K:=A+BK$ and $N':=N-\tau(t'')$. If, in addition, the robust DMPC optimization problem $\mathcal{P}_i$ has a feasible solution at time $t$, then it admits a feasible solution at time $t+1$, $t\in\mathbb{N}_{\geq 0}$.
\end{lem}

\begin{proof}
{An initial feasible solution to the optimization problem $\mathcal{P}_i$, $i\in\mathcal{V}$ is assumed to exist at time $t$, $t\in\mathbb{N}_{\geq 0}$ implying that the constraint $x_i^*(t+k|t)\in\hat{x}_i(t+k|t)\oplus \Delta$ in \eqref{whl6-eq:17e} holds, i.e., $x_i^*(t+k|t)\in{x}_i^*(t+k|t')\oplus \Delta$, with ${x}_i^*(t+k+1|t')=A{x}_i^*(t+k|t')+B{u}_i^*(t+k|t')$ and $t'\in\mathbb{N}_{[t-\bar{\tau},t-1]}$.
 
Note that the previous optimal predicted states $x_j^*(\cdot|t'')$ of agent $j$, $j\in\mathcal{N}_i$ are available for agent $i$, $i\in\mathcal{V}$ at time $t+1$, with $t''\in\mathbb{N}_{[t+1-\bar{\tau},t]}$. By \eqref{whl6-eq:9}, one gets $\hat{x}_j(t+1|t+1)=x_j^*(t+1|t'')$.} Thus, from $\tilde{x}_i(t+1|t+1)=x_i^*(t+1|t)$ and \eqref{whl6-eq:23}, we have
\begin{equation*}\label{whl6-eq:26}
\begin{aligned}
&\tilde{u}_i(t+1|t+1)\\
=&K\sum_{j\in\mathcal{N}_i}a_{ij}\big(\tilde{x}_i(t+1|t+1)-\hat{x}_j(t+1|t+1)\big)\\
&+\tilde{c}_i(t+1|t+1)\\
=&K\sum_{j\in\mathcal{N}_i}a_{ij}\big(\tilde{x}_i(t+1|t+1)-\hat{x}_j(t+1|t'')\\
&+\hat{x}_j(t+1|t'')-\hat{x}_j(t+1|t+1)\big)+\tilde{c}_i(t+1|t+1)\\
\end{aligned}
\end{equation*}
\begin{equation*}
\begin{aligned}
=&K\sum_{j\in\mathcal{N}_i}a_{ij}\big(x_i^*(t+1|t)-\hat{x}_j(t+1|t'')\big)+c_i^*(t+1|t'')\\
&+K\sum_{j\in\mathcal{N}_i}a_{ij}\big(\hat{x}_j(t+1|t'')-\hat{x}_j(t+1|t+1)\big)\\
=&K\sum_{j\in\mathcal{N}_i}a_{ij}\big(x_i^*(t+1|t'')-x_i^*(t+1|t'')+x_i^*(t+1|t)\\
&-\hat{x}_j(t+1|t'')\big)+c_i^*(t+1|t'')\\
&+K\sum_{j\in\mathcal{N}_i}a_{ij}\big(\hat{x}_j(t+1|t'')-\hat{x}_j(t+1|t+1)\big)\\
=&u_i^*(t+1|t'')+K\sum_{j\in\mathcal{N}_i}a_{ij}\big(\hat{x}_j(t+1|t'')-{x}_j^*(t+1|t''))\big)\\
&+K\sum_{j\in\mathcal{N}_i}a_{ij}\big(x_i^*(t+1|t)-x_i^*(t+1|t'')\big).
\end{aligned}
\end{equation*}

Since the constraint $x_i^*(t+k|t)\in{x}_i^*(t+k|t')\oplus \Delta$ holds, one then has $x_i^*(t+k|t)\in{x}_i^*(t+k|t'')\oplus \Delta$. Let $w_i'(t+1|t):=x_i^*(t+1|t)-\hat{x}_i(t+1|t)=x_i^*(t+1|t)-x_i^*(t+1|t'')$. From the definition of $w_i(t+k|t)=\sum_{j\in\mathcal{N}_i}a_{ij}\big(\hat{x}_j(t+k|t)-{x}_j^*(t+k|t)\big)$, we obtain
\begin{equation*}\label{whl6-eq:27}
\tilde{u}_i(t+1|t+1)=u_i^*(t+1|t'')+Kw_i'(t+1|t)+Kw_i(t+1|t'').
\end{equation*}

Applying the above control input, we have the predicted system state $\tilde{x}_i(t+2|t+1)$ as follows
\begin{equation}\label{whl6-eq:28}
\begin{aligned}
&\tilde{x}_i(t+2|t+1)\\
=&A\tilde{x}_i(t+1|t+1)+B\tilde{u}_i(t+1|t+1)\\
=&x_i^*(t+2|t'')+A_Kw_i'(t+1|t)+BKw_i(t+1|t'').
\end{aligned}
\end{equation}

Analogously, we have 
\begin{equation*}
\begin{aligned}
&\tilde{u}_i(t+2|t+1)\\
=&K\sum_{j\in\mathcal{N}_i}a_{ij}\big(\tilde{x}_i(t+2|t+1)-\hat{x}_j(t+2|t+1)\big)\\
&+\tilde{c}_i(t+2|t+1)\\
=&u_i^*(t+2|t'')+KBKw_i(t+1|t'')+Kw_i(t+2|t'')\\
&+KA_Kw_i'(t+1|t),
\end{aligned}
\end{equation*}
and
\begin{equation*}
\begin{aligned}
&\tilde{x}_i(t+3|t+1)\\
=&A\tilde{x}_i(t+2|t+1)+B\tilde{u}_i(t+2|t+1)\\
=&x_i^*(t+3|t'')+A_KBKw_i(t+1|t'')+BKw_i(t+2|t'')\\
&+A_K^2w_i'(t+1|t).
\end{aligned}
\end{equation*}

Then, the state sequence candidate $\tilde{x}_i(t+k|t+1)$ evolves in an iterative way according to 
\begin{equation}\label{whl6-eq:29}
\begin{aligned}
&\tilde{x}_i(t+1+k|t+1)\\
=&x_i^*(t+1+k|t'')+A_K^kw_i'(t+1|t)\\
&+\sum_{s=0}^{k-1}A_K^{k-1-s}BKw_i(t+1+s|t''),
\end{aligned}
\end{equation}
with $k\in\mathbb{N}_{[1,N')}$. 

Define the variable $r_i^k(t+1)\in\mathbb{R}^n$ for agent $i$, $i\in\mathcal{V}$ by
\begin{equation*}\label{whl6-eq:31}
\begin{aligned}
&{r}_i^k(t+1):=\\
&\begin{dcases}
w_i'(t+1|t),&k=0,\\
\sum_{s=0}^{k-1}A_K^{k-1-s}BKw_i(s)+A_K^kw_i'(t+1|t),&  k\in\mathbb{N}_{[1,N')},\\
\sum_{s=0}^{N'-1}A_K^{k-1-s}BKw_i(s)+A_K^kw_i'(t+1|t),&k\in\mathbb{N}_{[N',N)},
\end{dcases}
\end{aligned}
\end{equation*}
in which $w_i(s):=w_i(t+1+s|t'')$.

Considering $w_i(t+k|t'')\in\mathcal{W}$, $k\in\mathbb{N}_{[1,N']}$ and $w_i'(t+1|t)\in\mathcal{W}$, one gets the corresponding set $\mathcal{R}_i^k$
\begin{equation}\label{whl6-eq:32}
\mathcal{R}_i^k\
:=\begin{dcases}
\mathcal{W},&k=0,\\
\bigoplus_{s=0}^{k-1}A_K^{k-1-s}BK\mathcal{W}+A_K^k\mathcal{W},& k\in\mathbb{N}_{[1,N')},\\
\bigoplus_{s=0}^{N'-1}A_K^{k-1-s}BK\mathcal{W}+A_K^k\mathcal{W},& k\in\mathbb{N}_{[N',N)},
\end{dcases}
\end{equation}
with $r_i^k(t+1)\in\mathcal{R}_i^k$. 

Since the condition $\max_{k\in\mathbb{N}_{[1,N-1]}}\{\rho(\sum_{s=0}^{N'-1}A_K^{k-1-s}BK+A_K^k),\rho(\sum_{s=0}^{k-1}A_K^{k-1-s}BK+A_K^k)\}\leq 1$ holds, and by \eqref{whl6-eq:29} and \eqref{whl6-eq:32}, we get $\mathcal{R}_i^{k}\subseteq\Delta$, $k\in\mathbb{N}_{[1,N)}$. Hence, the constraint \eqref{whl6-eq:17e} at $t+1$ is satisfied.

Let $\bar{u}_i^*(t+k|t''):=K\sum_{j\in\mathcal{N}_i}a_{ij}\big(x_i^*(t+k|t'')-x_j^*(t+k|t'')\big)+c_i^*(t+k|t'')$ be the nominal optimal control input with $k\in\mathbb{N}_{[1,N'-1]}$. Note that all admissible disturbances satisfy $\forall w_i(t+k|t'')\in\Delta$. This implies that the nominal optimal control input satisfies $\bar{u}_i^*(t+k|t'')\in\overline{\mathcal{U}}_i$, where $\overline{\mathcal{U}}_i:=\mathcal{U}_i\ominus K\Delta$. It should be mentioned that the nominal control inputs $\bar{u}_i^*(\cdot|t'')$ are not implemented to the actual system, since $x_j^*(t+k|t'')$, $j\in\mathcal{N}_i$ is not available for agent $i$, $i\in\mathcal{V}$ at time $t''$. Next, substituting \eqref{whl6-eq:29} into \eqref{whl6-eq:24} yields
\begin{equation}
\tilde{u}_i(t+k|t+1)=\bar{u}_i^*(t+k|t'')+Kr_i^{k-1}(t+1),
\end{equation}
where $k\in\mathbb{N}_{[1,N']}$. 

Since the optimal control input sequence $\bm{u}_i^*(t'')$ is assumed to be feasible at time $t''$, we have $\bar{u}_i^*(t+k|t'')\in\overline{\mathcal{U}}_i$, $k\in\mathbb{N}_{[1,N']}$. Additionally, $r_i^{k-1}(t+1)\in\Delta$. When $k\in\mathbb{N}_{(N',N]}$, it follows from Assumption \ref{whl6-asm:2} that $\tilde{u}_i(t+k|t+1) \in\mathcal{U}_i$. Hence, at time $t+1$, from \eqref{whl6-eq:23} and \eqref{whl6-eq:24}, we further obtain that
\begin{equation}
\tilde{u}_i(t+k|t+1)\in \mathcal{U}_i,
\end{equation}
where $k\in\mathbb{N}_{[1,N]}$. This implies that the control input constraint in \eqref{whl6-eq:17d} holds.

Given the initial state $\tilde{x}_i(t+1|t+1)=x_i(t+1)$ and the control inputs defined in \eqref{whl6-eq:24}, one gets
\begin{equation}\label{whl6-eq:35}
\begin{aligned}
&\tilde{x}_i(t+1+N|t+1)\\
=&A\tilde{x}_i(t+N|t+1)+B\tilde{u}_i(t+N|t+1)\\
=&A\tilde{x}_i(t+N|t+1)\\
&+BK\sum_{j\in\mathcal{N}_i}a_{ij}\big(\tilde{x}_i(t+N|t+1)-\hat{x}_j(t+N|t+1)\big)\\
=&A\big({x}_i^*(t+N|t'')+r_i^{N-1}(t+1)\big)+BK\sum_{j\in\mathcal{N}_i}a_{ij}\\
&\big({x}_i^*(t+N|t'')+r_i^{N-1}(t+1)-{x}_j^*(t+N|t'')\big)\\
=&A{x}_i^*(t+N|t'')+A_Kr_i^{N-1}(t+1)\\
&+BK\sum_{j\in\mathcal{N}_i}a_{ij}\big({x}_i^*(t+N|t'')-{x}_j^*(t+N|t'')\big),
\end{aligned}
\end{equation}
in which $x_i^*(t+1+k|t''):=Ax_i^*(t+k|t'')+B{u}_i^*(t+k|t'')$, with $u_i^*(t+k|t''):=K\sum_{j\in\mathcal{N}_i}a_{ij}\big(x_i^*(t+k|t'')-x_j^*(t+k|t'')\big)$, $k\in\mathbb{N}_{[N',N]}$.

Let $\tilde{\mathbf{x}}(t+N|t+1):=\col(\tilde{x}_1(t+N|t+1),\dots,\tilde{x}_M(t+N|t+1))$ and $\mathbf{x}^*(t+N|t''):=\col({x}_1^*(t+N|t''),\dots,{x}_M^*(t+N|t''))$. Equipped with these notation, we rewrite \eqref{whl6-eq:35} in the following compact form
\begin{equation}\label{whl6-eq:36}
\begin{aligned}
&\tilde{\mathbf{x}}(t+1+N|t+1)\\
=&(I_M\otimes A+\mathcal{L}\otimes BK)\mathbf{x}^*(t+N|t'')\\
&+\big(I_M\otimes (A+ BK)\big)\mathbf{r}^{N-1}(t+1),
\end{aligned}
\end{equation}
where $\mathbf{r}^{N-1}(t+1)=\col(r_1^{N-1}(t+1),r_2^{N-1}(t+1),\dots,r_M^{N-1}(t+1))\in\mathbf{\Delta}$, with $\mathbf{\Delta}=\{\mathbf{r}^{N-1}\in\mathbb{R}^{Mn}\mid \|\mathbf{r}^{N-1}\|^2\leq M\eta^2\}$. Let the tightened terminal set be $\overline{\mathbf{X}}^f:=\mathbf{X}^f\ominus{\mathbf{\Delta}}$. When $\mathbf{x}^*(t+N|t'')\in\overline{\mathbf{X}}^f$, {Assumption}~\ref{whl6-asm:1} implies that $(I_M\otimes A+\mathcal{L}\otimes BK)\mathbf{x}^*(t+N|t'')\in\overline{\mathbf{X}}^f$. 

The corresponding tightened terminal set for agent $i$, $i\in\mathcal{V}$ becomes
\begin{equation}
\begin{aligned}
\overline{\mathcal{X}}_i^f:=&\mathcal{X}_i^f\ominus \Delta,
\end{aligned}
\end{equation}
with the constant $\eta>0$ and $\overline{\mathcal{X}}_i^f\neq \emptyset$.

 As a result, combining the above with \eqref{whl6-eq:35} and \eqref{whl6-eq:36}, one obtains
\begin{equation}
\begin{aligned}
\tilde{x}_i(t+1+N|t+1)&\in\overline{\mathcal{X}}_i^f+A_K\Delta\\
&\subseteq{\mathcal{X}}_i^f,
\end{aligned}
\end{equation}
which implies that the terminal set \eqref{whl6-eq:17f} is satisfied. By now, we have shown that $\tilde{\bm{c}}_i(t+1)$ is a feasible solution to the robust DMPC optimization problem $\mathcal{P}_i$ at time $t+1$.
\end{proof}
\begin{rem}
When the delay bound $\bar{\tau}$ increases, a larger prediction horizon $N$ is needed to construct the assumed predicted state sequence. As a result, the DMPC optimization problem becomes more computationally challenging. Additionally, it is hard to ensure the satisfaction of the condition \eqref{whl6-eq:21a}, which may lead to an infeasible solution. Therefore, one would choose suitable parameters to make a trade-off between the delay-tolerant capability and computational complexity.
\end{rem}

\begin{rem}
Note that the constraint in \eqref{whl6-eq:17e} must be satisfied by every optimized sequence $x_i^*(t+k|t)$ in the tube $\hat{x}_i(t+k|t)\oplus\Delta$, $t\in\mathbb{N}_{\geq 0}$. That is, the estimation errors take values in the specified set $\Delta$. From \eqref{whl6-eq:32}, a minimal robust invariant set $\mathcal{R}_i^\infty$ can be obtained following the method in \cite{mayne2005robust}. Also notice that $\Delta\subset\mathcal{R}_i^\infty$. Here we direct our attention to the prediction horizon $\mathbb{N}_{[0,N]}$, then replacing $\mathcal{R}_i^\infty$ with $\mathcal{R}_i^k$, $k\in\mathbb{N}_{[0,N]}$ yields a less conservative set. In particular, the predesigned consensus gain matrix $K$, the communication delay bound $\bar{\tau}$ and the prediction horizon $N$ should be suitably chosen such that the set $R_i^k$ satisfies $R_i^k\subseteq\Delta$. %It is natural to ask if the estimation error set $\Delta(t)$ can be time-varying. The answer is yes, and this set can be useful in more complex situation, such as the MAS with time-varying communication graph and the time-varying bounds on the estimation errors. 
\end{rem}
\subsection{Consensus analysis}
\label{whl6-sec:5-2}
In this subsection, we first provide two technical lemmas and then present the convergence analysis in \textbf{Theorem} \ref{whl6-thm:1}.

The following result on $c_i(t)$ is fundamental to the convergence analysis. 
\begin{lem}\label{whl6-lem:3}
Provided that the initial state ${x}_i(0)$ is feasible, the MAS in \eqref{whl6-eq:2} under the distributed consensus protocol ${u}_i^*(t|t)=K\sum_{j\in\mathcal{N}_i}a_{ij}\big(x_i^*(t|t)-\hat{x}_j(t|t)\big)+c_i^*(t|t)$, with $c_i(t)=c_i^*(t|t)$, satisfies the following property:
\begin{equation*}
\lim_{t\to\infty}c_i(t)=0.
\end{equation*}
\end{lem}
\begin{proof}
To prove the convergence of $c_i(t)$ as $t\to \infty$, we introduce the following function
\begin{equation*}
V_i(t):=J_i(\bm{c}_i^*(t))=\sum_{k=0}^{N-1}\|c_i^*(t+k|t)\|_{P_i}^2.
\end{equation*}

For the control candidate sequence $\tilde{\bm{c}}_i(t+1)$, we have
\begin{equation*}
\begin{aligned}
\tilde{V}_i(t+1)=&\sum_{k=0}^{N-1}\|\tilde{c}_i(t+1+k|t+1)\|_{P_i}^2\\
=&V_i(t)-\|c_i^*(t|t)\|_{P_i}^2.
\end{aligned}
\end{equation*}

The control candidate sequence $\tilde{\bm{c}}_i(t+1)$ is a feasible, but not necessarily optimal solution of the robust DMPC optimization problem $\mathcal{P}_i$ at $t+1$. Hence, it follows
\begin{equation*}
V_i(t+1)\leq \tilde{V}_i(t+1)=V_i(t)-\|c_i^*(t|t)\|_{P_i}^2.
\end{equation*}

Furthermore, one gets
\begin{equation}\label{whl6-eq:40}
V_i(t+1)-V_i(t)\leq-\|c_i^*(t|t)\|_{P_i}^2,
\end{equation}
which implies the Lyapunov function $V_i(t)$ is monotonically non-increasing as $t\to \infty$. Upon summing up $V_i(t+1)-V_i(t)$ in \eqref{whl6-eq:40}, we have
\begin{equation}
\begin{aligned}
&\lim_{k\to\infty}\sum_{t=0}^{k}(V_i(t+1)-V_i(t))\\
=&\lim_{k\to\infty}V_i(k+1)-V_i(0)\\
\leq&-\lim_{k\to\infty}\sum_{t=0}^k\|c_i^*(t|t)\|_{P_i}^2,
\end{aligned}
\end{equation}
and $V_i(t)$ as $t\to \infty$, satisfies
\begin{equation*}
0\leq V_i(\infty)\leq V_i(0)-\lim_{k\to\infty}\sum_{t=0}^k\|c_i^*(t|t)\|_{P_i}^2< \infty,
\end{equation*}
where $V_i(\infty):=\lim_{t\to\infty}V_i(t)$. Hence, it holds that
\begin{equation*}
\lim_{t\to\infty}\|c_i^*(t|t)\|_{P_i}^2=0.
\end{equation*}

Thus, $\lim_{t\to\infty}\|c_i(t)\|=0$ holds. By now, we have proved that the control variable $c_i(t)$ vanishes to zero as $t\to \infty$.
\end{proof}
The next lemma (see \cite[Lemma 7]{nedic2010constrained}) discussing the convergence property of the summable convolution sequence will be used in the proof of \textbf{Theorem} \ref{whl6-thm:1}.
\begin{lem}\label{whl6-lem:4}
For any given scalar $\beta\in(0,1)$, and the summable sequence $\{\alpha(t)\}$ satisfying $\lim_{t\to\infty}\alpha(t)=0$, it holds that $\lim_{k\to\infty}\sum_{t=0}^{k}\beta^{k-1-t}\alpha(t)=0$.
\end{lem}

Now, based on \textbf{Lemmas~\ref{whl6-lem:3}} and \textbf{\ref{whl6-lem:4}}, the consensus convergence analysis of the MAS is reported as follows.
\begin{thm}\label{whl6-thm:1}
For the MAS in \eqref{whl6-eq:2}, suppose that \textbf{Assumptions~\ref{whl6-asm:1}, \ref{whl6-asm:2}, \ref{whl6-asm:3}} and the condition in \textbf{Lemma \ref{whl6-lem:1}} hold. Then, the MAS converges to a neighborhood of the consensus set under the proposed robust DMPC consensus protocol.
\end{thm}
\begin{proof}
Applying the optimal control input $u_i^*(t|t)$ to the MAS in \eqref{whl6-eq:2} yields $x_i(t+1)=Ax_i(t)+Bu_i^*(t|t)$, with $x_i(t|t)=x_i(t)$. Then, we know that $x_i(t+1)=x_i^*(t+1|t)$. The closed-loop MAS in \eqref{whl6-eq:22} can be rewritten as
\begin{equation*}
\begin{aligned}
\mathbf{x}(t+1)=&(I_M\otimes A+\mathcal{L}\otimes BK)\mathbf{x}(t)+(I_M\otimes B)\mathbf{c}(t)\\
&+(I_M\otimes BK)\mathbf{w}(t),
\end{aligned}
\end{equation*}
where $\mathbf{c}(t):=\col(c_1(t),c_2(t),\dots,c_M(t))$. The average state of the MAS is defined by $\bar{x}(t):={1/M}(\bm{1}^\T\otimes I_n) \mathbf{x}(t)\in\mathbb{R}^n$, with $\bm{1}$ being an all-one vector of proper dimensions. Then,
\begin{equation}
\begin{aligned}
\bar{x}(t+1)=& \frac{1}{M}(\bm{1}^\T \otimes A)\mathbf{x}(t)+\frac{1}{M}(\bm{1}^\T \mathcal{L}\otimes BK)\mathbf{x}(t)\\
&+\frac{1}{M}(\bm{1}^\T \otimes B)\mathbf{c}(t)+\frac{1}{M}(\bm{1}^\T \otimes BK)\mathbf{w}(t)\\
=& \frac{1}{M}(\bm{1}^\T \otimes A)\mathbf{x}(t)+B\bar{c}(t)+BK\bar{w}(t)\\
=&  A\bar{x}(t)+B\bar{c}(t)+BK\bar{w}(t),
\end{aligned}
\end{equation}
with $\bar{c}(t):={1/M}(\bm{1}^\T\otimes I_n) \mathbf{c}(t)$ and $\bar{w}(t):={1/M}(\bm{1}^\T\otimes I_n) \mathbf{w}(t)$. Define further $\bar{\mathbf{c}}(t)=1/M\big((\bm{1}^\T\bm{1})\otimes I_m\big)\mathbf{c}(t)$ and $\bar{\mathbf{w}}(t)=1/M\big((\bm{1}^\T\bm{1})\otimes I_n\big)\mathbf{w}(t)$. Using notation $\xi_i(t):=x_i(t)-\bar{x}(t)$ and $\bm{\xi}:=\col(\xi_1,\xi_2,\dots,\xi_M)$, we have that
\begin{equation}\label{whl6-eq:43}
\begin{aligned}
&\bm{\xi}(t+1)\\
=&(I_M\otimes A+\mathcal{L}\otimes BK)\bm{\xi}(t)\\
&+(I_M\otimes B)I_{Mm}'\mathbf{c}(t)+(I_M\otimes BK)I_{Mn}'\mathbf{w}(t),
\end{aligned}
\end{equation}
with $I_{Mm}'=I_{Mm}-{1/M}\big((\bm{1}^\T\bm{1})\otimes I_m\big)$ and $I_{Mn}'=I_{Mn}-{1/M}\big((\bm{1}^\T\bm{1})\otimes I_n\big)$.

% The first method
Let $U:=[\bm{1}/\sqrt{M},U_2,\dots,U_M]\in\mathbb{R}^{M\times M}$ be an orthogonal matrix with $U_i^\T\mathcal{L}=\lambda_iU_i^\T$ and the Laplacian matrix can be diagonalized, that is
\begin{equation*}
U^\T \mathcal{L}U=\diag(0,\lambda_2,\dots,\lambda_M).
\end{equation*}

Using the property of Kronecker product, we obtain that
\begin{equation*}
\begin{aligned}
&(U^\T\otimes I_n)(I_M\otimes A+\mathcal{L}\otimes BK)(U\otimes I_n)\\
=&\diag(A,A+\lambda_2BK,\dots,A+\lambda_MBK).
\end{aligned}
\end{equation*}

Define $\tilde{\bm{\xi}}(t)=\col(\tilde{\xi}_1(t),\tilde{\xi}_2(t),\dots,\tilde{\xi}_M(t)):=(U^\T \otimes I_n)\bm{\xi}(t)$. Then \eqref{whl6-eq:43} can be expressed as
\begin{equation}
\begin{aligned}
&\tilde{\bm{\xi}}(t+1)\\
=&\diag(A,A+\lambda_2BK,\dots,A+\lambda_MBK)\tilde{\bm{\xi}}(t)\\
&+(U^\T \otimes I_n)(I_M\otimes B)I_{Mm}'\mathbf{c}(t)\\
&+(U^\T \otimes I_n)(I_M\otimes BK)I_{Mn}'\mathbf{w}(t).
\end{aligned}
\end{equation}

% The second method
Next, we define the transition matrix $\Phi:=\diag(A,A+\lambda_2BK,\dots,A+\lambda_MBK)$, $\mathcal{B}:=(U^\T \otimes I_n)(I_M\otimes B)I_{Mm}'$, and $\mathcal{B}_k:=(U^\T \otimes I_n)(I_M\otimes BK)I_{Mn}'$, then \eqref{whl6-eq:43} becomes
\begin{equation}
\tilde{\bm{\xi}}(t+1)=\Phi\tilde{\bm{\xi}}(t)+\mathcal{B}\mathbf{c}(t)+\mathcal{B}_k\mathbf{w}(t),
\end{equation}
and thus
\begin{equation*}
\tilde{\bm{\xi}}(t)=\Phi^{t}\tilde{\bm{\xi}}(0)+\sum_{k=0}^{t-1}\Phi^k\mathcal{B}\mathbf{c}(t-1-k)+\sum_{k=0}^{t-1}\Phi^k\mathcal{B}_k\mathbf{w}(t-1-k),
\end{equation*}
with $t\in\mathbb{N}_{\geq 1}$. From the definition of $\tilde{\bm{\xi}}(t)=(U^\T \otimes I_n)\bm{\xi}(t)$, it is easy to know that $\tilde{\xi}_1(t)=1/\sqrt{M}\big(\sum_{i=1}^M\xi_i(t)\big)=0$. Also, 
By invoking \textbf{Lemma~\ref{whl6-lem:1}} (i.e., $\rho(A+\lambda_iBK)<1$, $i=2,3,\dots,M$), we get 
\begin{equation}\label{whl6-eq:47}
\lim_{t\to\infty}\Phi^{t}\tilde{\bm{\xi}}(0)=0.
\end{equation}

From $I_{Mm}'=I_{Mm}-{1/M}\big((\bm{1}^\T\bm{1})\otimes I_m\big)=\big(I_{M}-{1/M}(\bm{1}^\T\bm{1})\big)\otimes I_m$, one gets the matrix $I_{Mm}'$ has eigenvalue $0$ with multiplicity $m$ and eigenvalue $1$ with multiplicity $Mm-m$. Here, it is obtained that $\rho(I_{Mm}')\leq 1$. In addition, $\rho(A+\lambda_iBK)\leq 1$, one gets $I_{Mm}'(\rho(A+\lambda_iBK))^k< 1$, $k\in\mathbb{N}_{\geq 1}$. In light of this, there always exists a constant $\beta\in(0,1)$ such that
\begin{equation}\label{whl6-eq:48}
\|I_{Mm}'\Phi^k\|\leq\beta^k<1.
\end{equation}

Let $E(t-1-k)=\|(U^\T \otimes I_n)\|\|I_M\otimes B\|\|\mathbf{c}(t-1-k)\|$. 

By application of the Cauchy-Schwarz inequality and the inequality in \eqref{whl6-eq:48}, we have

\begin{equation}
\begin{aligned}
&\lim_{t\to\infty}\sum_{k=0}^{t-1}\Phi^k\mathcal{B}\mathbf{c}(t-1-k)\\
\leq&\lim_{t\to\infty}\sum_{k=0}^{t-1}\|\Phi^k\mathcal{B}\mathbf{c}(t-1-k)\|\\
\leq&\lim_{t\to\infty}\sum_{k=0}^{t-1}\|I_{Mm}'\Phi^k\|\cdot\|(U^\T \otimes I_n)\|\\
&\cdot\|I_M\otimes B\|\cdot\|\mathbf{c}(t-1-k)\|\\
\leq&\lim_{t\to\infty}\sum_{k=0}^{t-1}\beta^kE(t-1-k).
\end{aligned}
\end{equation}

Further, according to \textbf{Lemmas~\ref{whl6-lem:3}} and \textbf{\ref{whl6-lem:4}}, we obtain
\begin{equation}\label{whl6-eq:51}
\lim_{t\to\infty}\sum_{k=0}^{t-1}\Phi^k\mathcal{B}\mathbf{c}(t-1-k)=0.
\end{equation}

From the definition of set $\mathcal{R}_i^\infty$, we know that
\begin{equation}\label{whl6-eq:52}
\lim_{t\to\infty}\sum_{k=0}^{t-1}\Phi^k\mathcal{B}_k\mathbf{w}(t-1-k)=\bm{R}^\infty,
\end{equation}
where $\bm{R}^{\infty}=\col(\mathcal{R}_1^{\infty},\mathcal{R}_2^{\infty},\dots,\mathcal{R}_M^{\infty})$. Combining \eqref{whl6-eq:47}, \eqref{whl6-eq:51} and \eqref{whl6-eq:52} gives
\begin{equation}
\lim_{t\to\infty}\tilde{\bm{\xi}}(t)\in\bm{R}^{\infty},
\end{equation}
which implies that the MAS converges to a neighborhood of the consensus set. The proof is completed. 
\end{proof}

It is shown that the sequence $c_i(t)$ generated by solving the online DMPC optimization problem converges to zero in \textbf{Lemma \ref{whl6-lem:3}}. Then, we prove the consensus convergence of the constrained MAS with bounded time-varying delays controlled by the proposed consensus protocol in \textbf{Theorem \ref{whl6-thm:1}}. It is worth mentioning that the results can be extended to the cases where the MAS are affected by unknown additive disturbances.

\section{Numerical examples}
\label{whl6-sec:6}
In this section, two numerical examples are conducted to verify the effectiveness of the developed theoretical results.

\noindent\textbf{Example 1}: \emph{Linear MAS with semi-stable dynamics}\\
Consider the MAS consisting of five discrete-time linear systems and each agent $i$ is described by \cite{li2018receding}
\begin{equation}
x_i(t+1)=Ax_i(t)+Bu_i(t), \ i=1,2,\dots,5,
\end{equation}
with
\begin{equation*}
A=\begin{bmatrix}
0.8& 0.1& 0.1& 0& 0\\
0& 0.9 &0& 0.1& 0\\
0.1 &0.1 &0.6& 0.1 &0.1\\
0 &0.1 &0.1& 0.8 &0\\
0.1 &0.1 &0& 0 &0.8\\
\end{bmatrix},\ B=\begin{bmatrix}
-0.1 &0.1\\
0.1&-0.2\\
0&-0.3\\
0.08&0.1\\
0.2&0.08\\
\end{bmatrix}.
\end{equation*}

The control input constraints of five agents are all set as $\|u_i\|_{\infty}\leq0.3$. The initial states of five agents are set as $x_1(0)=[0.94;1.22;-1.12; 1.24; 0.40]$, $x_2(0)=[1.21; -0.66;$ $0.14; 1.37; 1.39]$, $x_3(0)=[-1.03; 1.41; 1.37; -0.04; 0.90]$, $x_4(0)=[-1.07;- 0.23; 1.25; 0.88; 1.38]$ and $x_5(0)=[0.47; $ $-1.39;1.05; 1.30; 0.54]$, respectively. The maximum communication delay $\bar{\tau}=3$. The Laplacian matrix concerning the communication graph $\mathcal{G}$ is
\begin{equation*}
\mathcal{L}=\begin{bmatrix}
1& -0.5&0&-0.5&0\\
-0.5&1&-0.5&0&0\\
0&-0.5&1&0&-0.5\\
-0.5&0&0&1&-0.5\\
0&0&-0.5&-0.5&1\\
\end{bmatrix}.
\end{equation*}

The predesigned consensus gain matrix is chosen as $K = [0.1258,-0.1015, 0.0542, 0.0071,-2.443; -0.0787, 0.1863,$ $ 0.0637, 0.0982,-0.1376]$ following the method in \cite{li2018receding}. The prediction horizon is selected as $N=10$, and the estimation error set $\Delta=\{\delta\mid\|\delta\|\leq0.3\}$. Let $\epsilon=\sqrt{60}$. The weighting matrix ${P}_i=I_2$ and\\ 
 $S =\begin{bmatrix}
2.551&-0.447& 0.119&-0.813&-1.069\\
-0.447& 4.028& 0.227&1.356& -2.664\\
 0.119& 0.227& 1.799& 0.74&-2.431\\
-0.813& 1.356&0.74& 3.884&-3.689\\
-1.069&-2.664&-2.431&-3.689& 10.081\\
\end{bmatrix}$.

\begin{figure}[!ht]
\includegraphics[width=1\columnwidth]{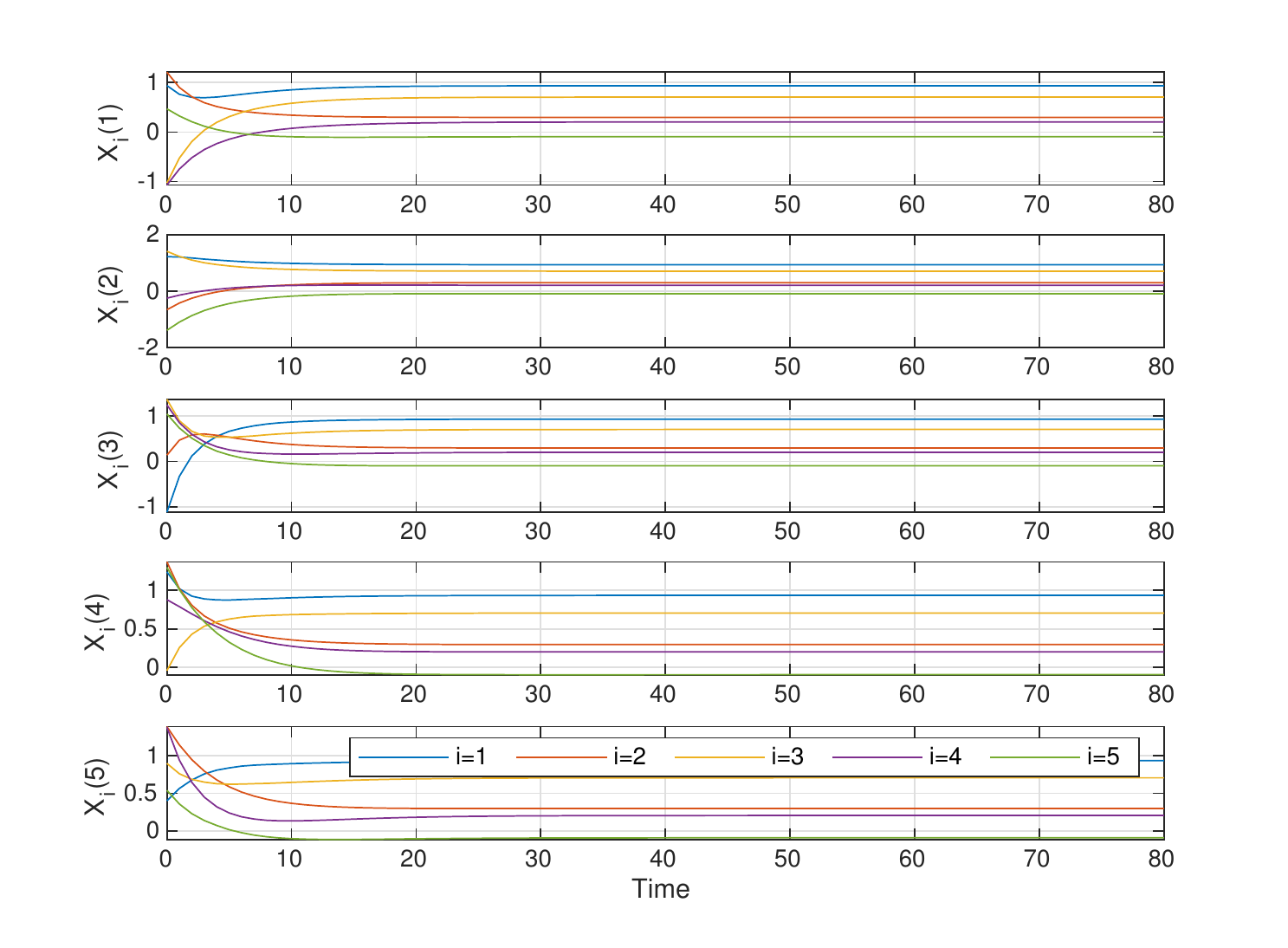}
\caption{States $x_i(t)$ of five agents under the consensus protocol in \cite{li2018receding}. }
\label{whl6-fig:3}
\end{figure}

\begin{figure}[!ht]
\includegraphics[width=1\columnwidth]{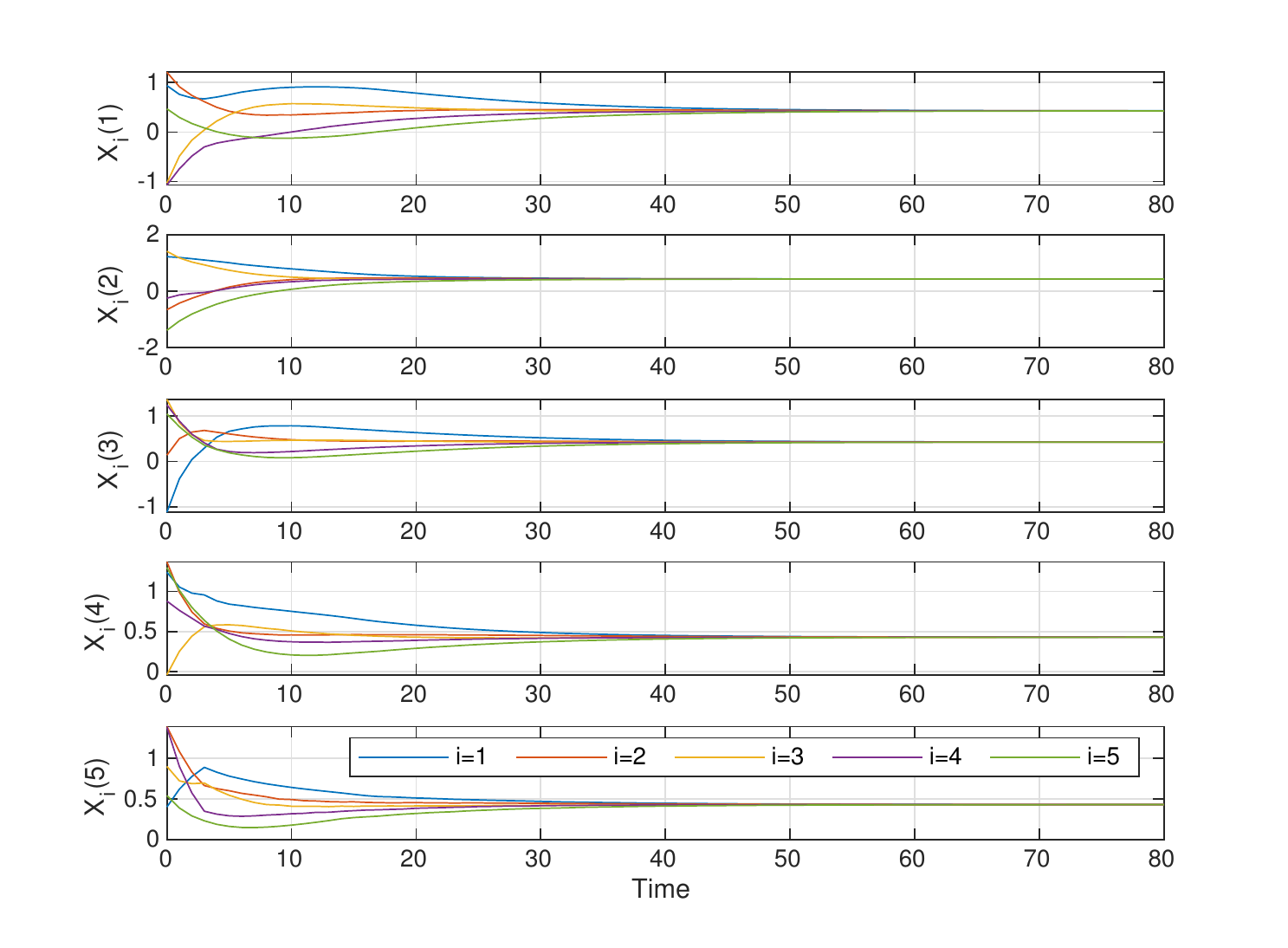}
\caption{States $x_i(t)$ of five agents under the proposed protocol.}
\label{whl6-fig:4}
\end{figure}

\begin{figure}[!ht]
\centering
  \includegraphics[clip,width=1\columnwidth]{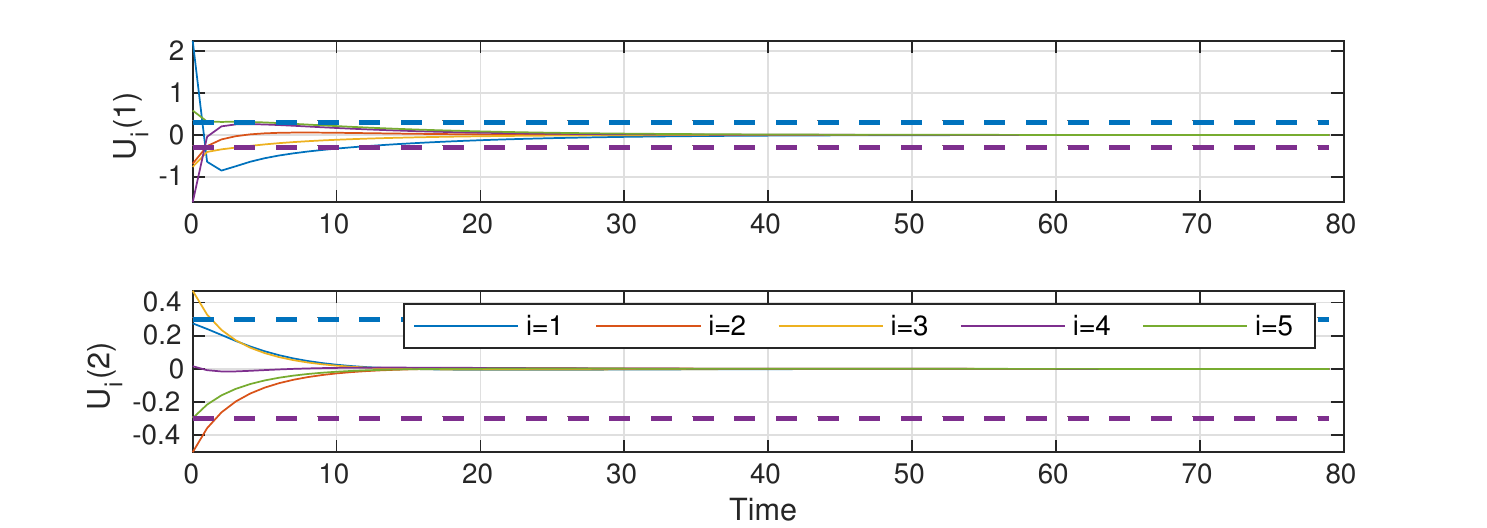}
\caption{Control inputs $\hat{\varkappa}_i(t)$ of five agents under the predesigned consensus protocol.}\label{whl6-fig:5}
\end{figure}

\begin{figure}[!ht]
\centering
  \includegraphics[clip,width=1\columnwidth]{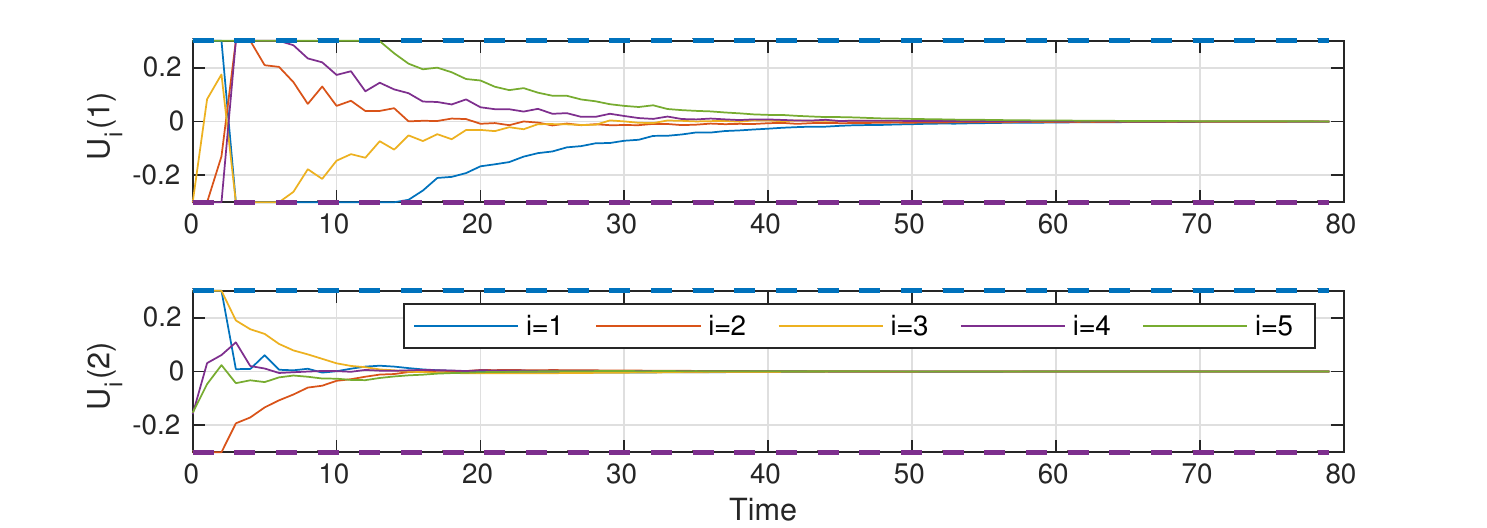}
\caption{Control inputs $u_i(t)$ of five agents under the robust DMPC-based consensus protocol.}\label{whl6-fig:6}
\end{figure}

\begin{figure}[!ht]
\centering
  \includegraphics[clip,width=1\columnwidth]{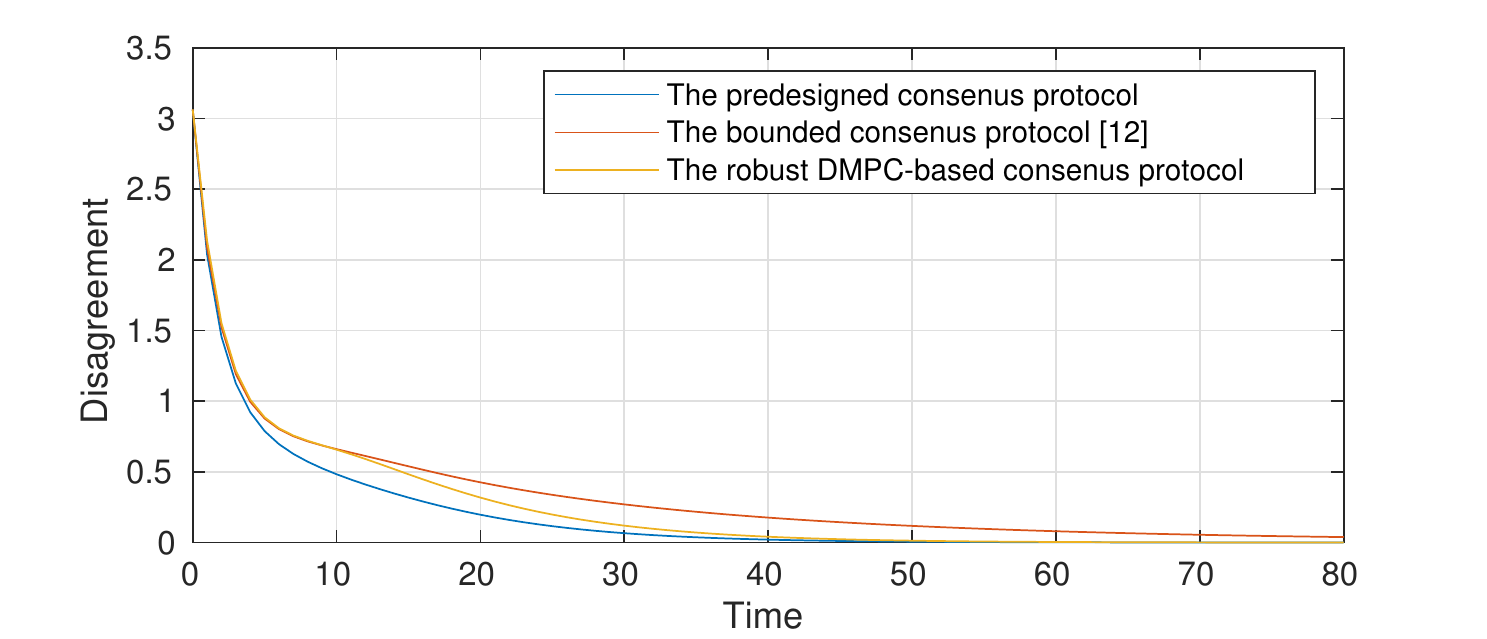}
\caption{Consensus performance comparison of different consensus protocols.}\label{whl6-fig:6a}
\end{figure}

Figs.~\ref{whl6-fig:3} and \ref{whl6-fig:4} show the states of five agents with communication delays under the consensus protocol in \cite{li2018receding} and the robust DMPC-based consensus protocol. Clearly, the consensus is not achieved under the protocol in \cite{li2018receding}, due to time-varying delays among agents. It can be seen from Fig.~\ref{whl6-fig:4} that five agents reach consensus under the proposed protocol.%, which is in line with Theorem \ref{whl6-thm:1}. 

On the other hand, the control inputs of all agents under the predesigned consensus protocol and the robust DMPC-based consensus protocol are shown in Figs.~\ref{whl6-fig:5} and \ref{whl6-fig:6}. It can be observed that the MAS only under the proposed consensus protocol satisfies control input constraints. 

To quantify the consensus performance, we define a disagreement function $D(t)=\sum_{i=1}^M\sum_{j\in\mathcal{N}_i}a_{ij}\|x_i(t)-x_j(t)\|/M$ for the MAS under different consensus protocols. The bounded consensus protocol from \cite{yang2014global} is given as 
\begin{equation*}
\begin{aligned}
u_i=\begin{dcases}
0.3,&u_i\geq 0.3,\\
-0.3,&u_i\leq -0.3,\\
K'\sum_{j\in\mathcal{N}_i}a_{ij}(x_i(t)-x_j(t)),&-0.3< u_i< 0.3,\\
\end{dcases}
\end{aligned}
\end{equation*}
in which $K'=[0.0846, -0.1523, 0.0028, -0.1044, -0.2256;\\-0.0818, 0.2567, 0.2256, -0.0423, -0.0479]$. As illustrated in Fig. \ref{whl6-fig:6a}, all disagreement trajectories monotonically decrease. The MAS under the predesigned consensus protocol converges at a faster rate than the proposed robust DMPC-based consensus protocol and the bounded consensus protocol in \cite{yang2014global}. However, the control input constraints cannot be guaranteed under the predesigned consensus protocol as reported in Fig.~\ref{whl6-fig:5}. In contrast, the other two types of consensus protocols guarantee the satisfaction of control input constraints. In particular, the proposed DMPC-based consensus protocol achieves a faster consensus convergence rate than the bounded consensus protocol in \cite{yang2014global}.

\noindent\textbf{Example 2}: \emph{Linear MAS with unstable dynamics}\\ 
Consider the MAS consisting of four discrete-time oscillators and each agent $i$ is characterized by \cite{nguyen2015sub}
\begin{equation}\label{whl6-eq:54}
x_i(t+1)=Ax_i(t)+Bu_i(t), \ i=1,2,3,4,
\end{equation}
with 
\begin{equation*}
A=\begin{bmatrix}
0& 1\\
-1.15&0\\
\end{bmatrix},\ B=\begin{bmatrix}
0.5\\
0.5\\
\end{bmatrix}.
\end{equation*}

The control input constraints of four agents are all set as $\|u_i\|_{\infty}\leq0.1$. The initial states of four agents are set as $x_1(0)=[-0.18; 0.21]$, $x_2(0)=[0.32; -0.18]$, $x_3(0)=[-0.29; -0.14]$ and $x_4(0)=[-0.22; 0.24]$, respectively. The maximum communication delay $\bar{\tau}=2$. The Laplacian matrix of $\mathcal{G}$ is
\begin{equation*}
\mathcal{L}=\begin{bmatrix}
1& -0.5&0&-0.5\\
-0.5&1&-0.5&0\\
0&-0.5&1&-0.5\\
-0.5&0&-0.5&1\\
\end{bmatrix}.
\end{equation*}

The prediction horizon is $N=7$ and the estimation error set is $\Delta=\{\delta\mid\|\delta\|\leq0.1\}$. Let $\epsilon=\sqrt{0.96}$. The weighting matrix $P_i=50$, the predesigned consensus gain matrix is designed as $K = [0.2748,-0.3148]$ and $S =\begin{bmatrix}
4.4733  &  0.8746\\
0.8746   & 3.3690\\
\end{bmatrix}$. 

Figs.~\ref{whl6-fig:7} and \ref{whl6-fig:8} show the states of four agents with communication delays using the consensus protocol in \cite{li2018receding} and the proposed robust DMPC-based consensus protocol. It can be observed from Fig.~\ref{whl6-fig:7} that unstable agents with communication delays cannot reach consensus under the protocol in \cite{li2018receding}. In contrast, as shown in Fig. \ref{whl6-fig:8} four agents with communication delays asymptotically converge to consensus. The simulation results verify that the proposed consensus protocol applies to the general linear constrained MAS with semi-stable and unstable dynamics.

In addition, as shown in Fig.~\ref{whl6-fig:9}, the control inputs of agents under the predesigned consensus protocol cannot satisfy control input constraints. In contrast, Fig.~\ref{whl6-fig:10} exhibits that all agents under the proposed consensus protocol satisfy control input constraints. 

\begin{figure}[!ht]
\includegraphics[width=1\columnwidth]{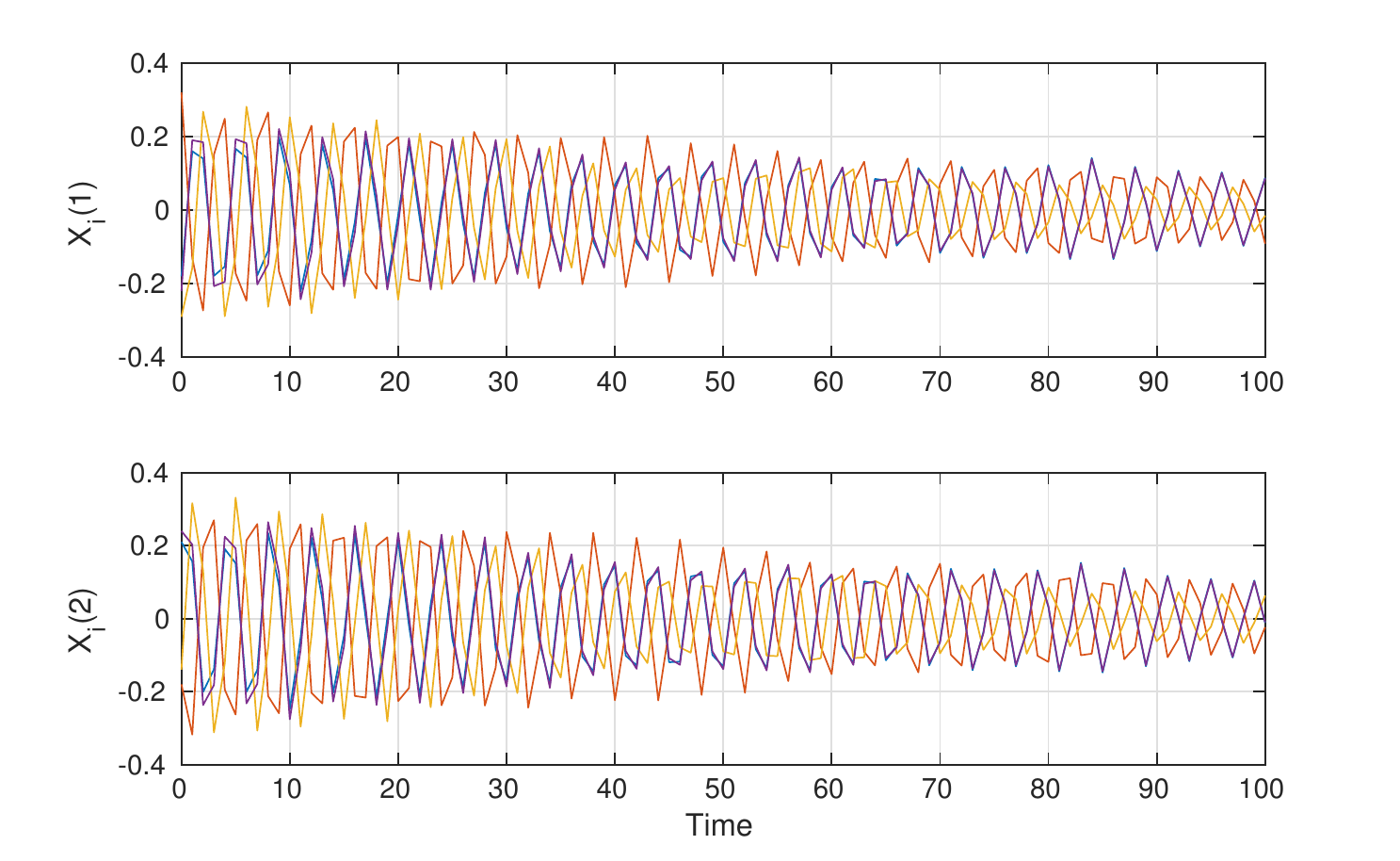}
\caption{States $x_i(t)$ of four agents under the consensus protocol in \cite{li2018receding}. }
\label{whl6-fig:7}
\end{figure}

\begin{figure}[!ht]
\includegraphics[width=1\columnwidth]{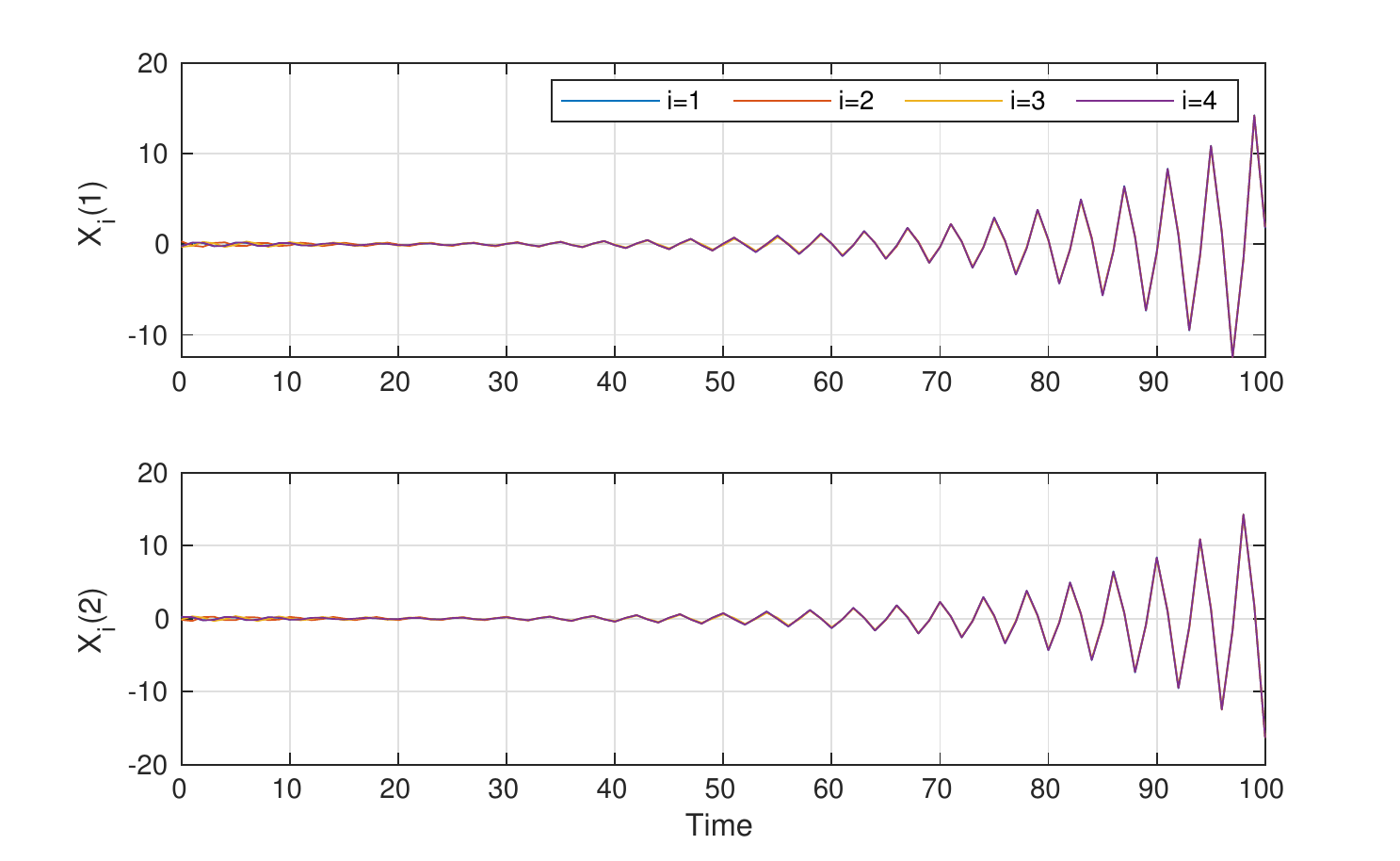}
\caption{States $x_i(t)$ of four agents under the proposed protocol.}
\label{whl6-fig:8}
\end{figure}

\begin{figure}[!ht]
\centering
  \includegraphics[clip,width=1\columnwidth]{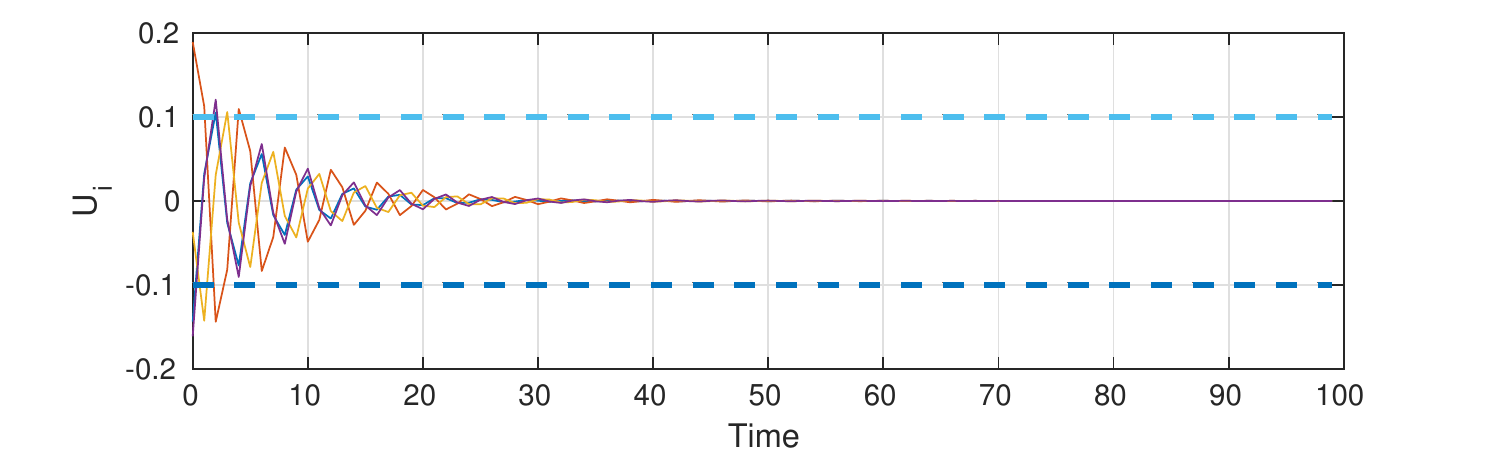}
\caption{Control inputs $\hat{\varkappa}_i(t)$ of four agents under the predesigned consensus protocol.}\label{whl6-fig:9}
\end{figure}

\begin{figure}[!ht]
\centering
  \includegraphics[clip,width=1\columnwidth]{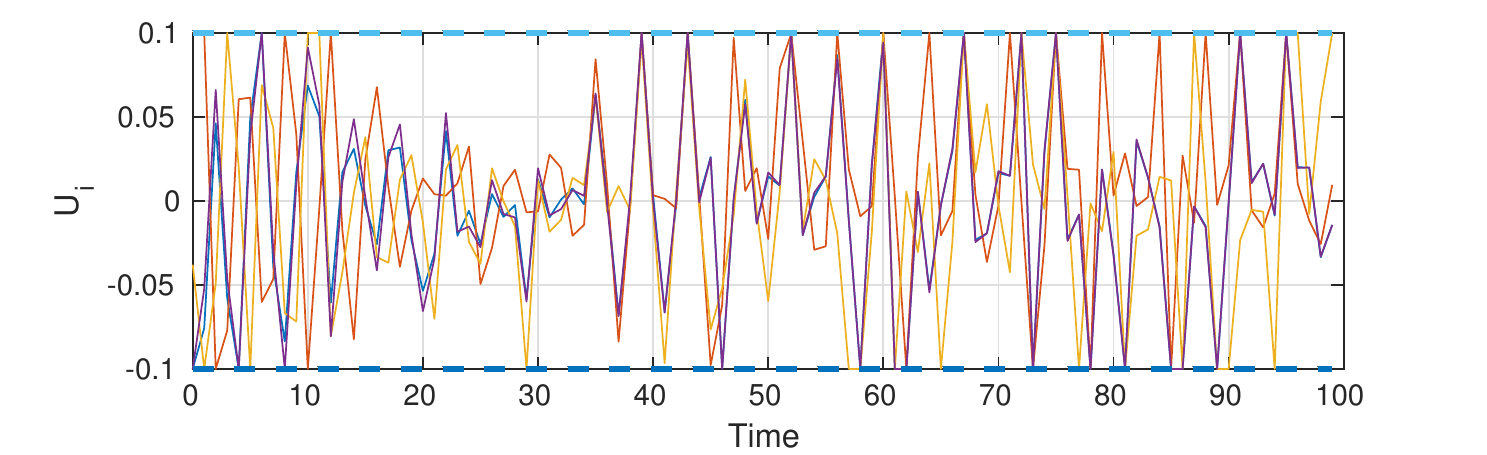}
\caption{Control inputs $u_i(t)$ of four agents under the robust DMPC-based consensus protocol.}\label{whl6-fig:10}
\end{figure}
%%%%%%%%%%%%%%%%%%%%%%%%%%%%%%%%%%%%%%%%%%%%%%%%%%%%%%%%%%%%%%%%%%%%%%%%%%%%%
        %%%%%%%%%%%%%%%%%%%%%%% Section V %%%%%%%%%%%%%%%%%%%%%%%%%%%%%
%%%%%%%%%%%%%%%%%%%%%%%%%%%%%%%%%%%%%%%%%%%%%%%%%%%%%%%%%%%%%%%%%%%%%%%%%%%%%
\section{Conclusion}
\label{whl6-sec:7}
This work designed a robust DMPC-based consensus framework for the general linear MAS with input constraints and bounded time-varying communication delays. The proposed distributed consensus protocol achieved suboptimal consensus performance using inverse optimal control and tube-based MPC techniques while satisfying control input constraints. Furthermore, the recursive feasibility of the DMPC optimization problem and the consensus convergence of the constrained MAS were rigorously analyzed. Two numerical examples were given to verify the proposed distributed consensus protocol. 

Future work will generalize the distributed consensus framework to address the time-varying and unbalanced communication networks among the MAS and will consider the privacy and resilience issues of the constrained MAS under cyber-attacks.
% \section*{Acknowledgment}
\bibliographystyle{IEEEtran}
\bibliography{whl6coupling}
% \begin{IEEEbiography}{Changxin Liu}
% received the Ph.D. degree in mechanical engineering at the University of Victoria, Victoria, BC, Canada, in 2021. He is currently a Postdoctoral Researcher with the School of Electrical Engineering and Computer Science, KTH Royal Institute of Technology, Stockholm, Sweden. His research interests focus on distributed optimization and control of networked systems. He is an active reviewer for more than 10 international journals and conferences, and was an Outstanding Reviewer for \emph{IEEE Transactions on Cybernetics} in 2018.
% \end{IEEEbiography}
\begin{IEEEbiographynophoto}{Henglai Wei} received his M.Sc. degree in control theory from Northwestern Polytechnical University, Xi'an, China, in 2017, and the Ph.D. degree in mechanical engineering from the University of Victoria, Victoria, BC, Canada, in 2022. He is currently a Postdoctoral Researcher with the Faculty of Engineering and Computer Science, University of Victoria. His current research interests include model predictive control and distributed control and optimization of intelligent systems.
\end{IEEEbiographynophoto}
\begin{IEEEbiographynophoto}{Changxin Liu}
received the Ph.D. degree in mechanical engineering at the University of Victoria, Victoria, BC, Canada, in 2021. He is currently a Postdoctoral Researcher with the School of Electrical Engineering and Computer Science, KTH Royal Institute of Technology, Stockholm, Sweden. His research interests focus on distributed optimization and control of networked systems. He is an active reviewer for more than 10 international journals and conferences, and was an Outstanding Reviewer for \emph{IEEE Transactions on Cybernetics} in 2018.
\end{IEEEbiographynophoto}

\begin{IEEEbiographynophoto}{Yang Shi}
received the Ph.D. degree in electrical and computer engineering from the University of Alberta, Edmonton, AB, Canada, in 2005. From 2005 to 2009, he was an Assistant Professor and an Associate Professor in the Department of Mechanical Engineering, University of Saskatchewan, Saskatoon, SK, Canada. In 2009, he joined the University of Victoria, Victoria, BC, Canada, where he is currently a Professor in the Department of Mechanical Engineering. He was a Visiting Professor with the University of Tokyo, Tokyo, Japan, in 2013. His current research interests include networked and distributed systems, model predictive control (MPC), cyber-physical systems (CPS), robotics and mechatronics, navigation and control of autonomous systems (AUV and UAV), and energy system applications.

Prof. Shi was a recipient of the University of Saskatchewan Student Union Teaching Excellence Award in 2007; the Faculty of Engineering Teaching Excellence Award in 2012 and the Craigdarroch Silver Medal for Excellence in Research in 2015 from the University of Victoria; the 2017 \emph{IEEE Transactions on Fuzzy Systems} Outstanding Paper Award for his coauthored paper; the JSPS Invitation Fellowship (short-term); and the Humboldt Research Fellowship for Experienced Researchers in 2018. He has been a member of the IEEE IES Administrative Committee since 2017 and is currently the Chair of IEEE IES Technical Committee on Industrial Cyber-Physical Systems. He is the Co-Editor-in-Chief of \emph{IEEE Transactions on Industrial Electronics}, and also serves as an Associate Editor for \emph{Automatica}, \emph{IEEE Transactions on Automatic Control}, and \emph{IEEE Transactions on Cybernetics}. He is a fellow of IEEE, ASME, Engineering Institute of Canada (EIC), and Canadian Society for Mechanical Engineering (CSME), and a registered Professional Engineer in British Columbia, Canada.
\end{IEEEbiographynophoto}
% % insert where needed to balance the two columns on the last page with
% %\newpage

\end{document}